\documentclass[reqno]{amsart}%
\usepackage{amsmath}
\usepackage{amssymb}
\usepackage{amsfonts}%
\usepackage{graphicx}

\providecommand{\U}[1]{\protect\rule{.1in}{.1in}}
\newtheorem{theorem}{Theorem}
\theoremstyle{plain}

\newtheorem{corollary}{Corollary}

\newtheorem{definition}{Definition}

\newtheorem{notation}{Notation}

\newtheorem{property}{Property}
\newtheorem{proposition}{Proposition}

\numberwithin{equation}{section}
\begin{document}
\title[Entity-oriented spatial coding]{Entity-oriented spatial coding and discrete topological spatial relations}

\begin{abstract}
Based on a newly proposed spatial data model -- spatial chromatic model (SCM),
we developed a spatial coding scheme, called full-coded ordinary arranged
chromatic diagram (full-OACD). Full-OACD is a type of spatial tessellation,
where space is partitioned into a number of subspaces such as cells, edges,
and vertexes. These subspaces are called spatial particles and assigned with
unique codes -- chromatic codes. The generation, structures, computations, and
properties of full-OACD are introduced and relations between chromatic codes
and particle spatial topology are investigated, indicating that chromatic
codes provide a potential useful and meaningful tool not only for spatial
analysis in geographical information science, but also for other relevant
disciplines such as discrete mathematics, topology, and computer science.

\end{abstract}
\author{Weining Zhu}
\address[Zhu, W. N.]{ Laboratory of Geospatial Information, Zhejiang University}
\email{zhuwn@zju.edu.cn}
\maketitle

\section{Introduction}

Coding the objects has been widely used in many scientific and technological
fields, such as telecommunications, bioinformatics, and computer cryptography,
in which information has been expressed, transferred, and interpreted by
various codes in numbers, strings, or symbols. In geographic information
science (GIS), there are also some relevant applications of coding. For
example, a geographical coordinate systems provides a coding scheme using a
single or a series of coordinates to represent a spatial entity or region
\cite{Longley2001}, \cite{Hardisty2010}. Spatial index assigns codes (indexes)
to spatial objects so that they can be rapidly retrieved from spatial
databases \cite{Schubert2013}, \cite{Halaoui2008}. In geocoding systems, land
lots and zip codes allow spatial locations and postal addresses to be readily
memorized and exclusively identified \cite{Curry1998}, \cite{Davis2007}.

The objective of this study is to do the similar work for coding the pure
space itself. Actually a plannar Cartesian coordinate system is also a coding
scheme where a point in space is coded by such as a coordinate ($x,y$). Based
on a newly proposed GIS data model -- spatial chromatic model (SCM)
\cite{Zhu2015}, we suggest a spatial coding scheme, called full-coded ordinary
arranged chromatic diagram (full-OACD). Full-OACD can be taken as an extension
of OACD, which is a standard pattern of SCM. SCM has demonstrated its
significant potentials for GIS theories and applications in diverse aspects:
the first law of geography, reasoning spatial topology, point pattern
recognition, and generalized Voronoi algorithms, etc. \cite{Zhu2010},
\cite{Zhu2015}.

Space in SCM is defined as the object-oriented space where the elementary unit
is a cell. A cell is characterized by its chromatic code, typically a string
of natural numbers. One problem of OACDs is that only cells are coded, but
cellular boundaries and feature nodes, such as edges and vertexes generated
from half-plane partitions, have not been coded, and hence we may lose some
particular spatial information, for example, the subspaces somewhere that are
unable to be assigned to any cell. To solve this problem, we therefore
extended OACD to full-OACD, a full-space coding scheme. In full-OACD, all
spatial components, including cells, edges and vertexes, are coded in a
spatially and mathematically consistent way.

The below sections will introduce, analyze, and discuss the procedures of
generating full-OACDs, some important definitions, notations, properties, and
theorems (Section 2), topological relations among cells, edges, vertexes, and
complexes (Section 3), as well as their spatial implications, notes, and
suggested future work (Section 4).

\section{Full-coded ordinary arranged chromatic diagram}

Let $\boldsymbol{P}=\{p_{1},p_{2},\ldots,p_{n}\}$ is a point set containing
$n$ points associated with an index set $\boldsymbol{I}=\{1,2,\ldots,n\}$. The
point set is also called the \emph{generator set} and points in
$\boldsymbol{P}$ are \emph{generators}, which can be treated as geographical
entities or just any general objects. The set $\boldsymbol{Q}$ is a family of
subsets of $\boldsymbol{P}$ consisting of all unordered point-pairs in
$\boldsymbol{P}$, that is, $\boldsymbol{Q}=\{\{p_{i},p_{j}\}|p_{i},p_{j}%
\in\boldsymbol{P}$, $i\neq j$, $i,j\in\boldsymbol{I}\}$. The generation of a
full-OACD follows the below steps.

\begin{description}
\item[Step (1)] With respect to a point-pair $q=\{p_{i},p_{j}\}\in
\boldsymbol{Q}$, using their perpendicular bisector $pb\langle i,j\rangle$ to
partition the space into two half-planes $hp(i,j)$ and $hp(j,i)$, where a
point $p$ in $hp(i,j)$ is with Euclidean distance $d(p,p_{i})<d(p,p_{j})$, in
$hp(j,i)$ with $d(p,p_{j})<d(p,p_{i})$, and in $pb\langle i,j\rangle$ with
$d(p,p_{i})=d(p,p_{j})$.

\item[Step (2)] Assign two half-planes $hp(i,j)$ and $hp(j,i)$ the codes
$(p_{1}^{0}$, $p_{2}^{0}$, $\ldots$, $p_{i}^{1}$, $\ldots$, $p_{n}^{0})$ and
$(p_{1}^{0}$, $p_{2}^{0}$, $\ldots$, $p_{j}^{1}$, $\ldots$, $p_{n}^{0}) $,
respectively, in which the subscript number corresponds to the index of each
point, and the superscript number is the assigned numerical variable $t(q)$.
In this way, only for points $p_{i}$ or $p_{j}$, $t(q)=1$, but for the others,
$t(q)=0$. Similarly, assign the bisector $pb\langle i,j\rangle$ with code
$(p_{1}^{0}$, $p_{2}^{0}$, $\ldots$, $p_{i}^{\frac{1}{2}}$, $\ldots$,
$p_{j}^{\frac{1}{2}}$, $\ldots$, $p_{n}^{0})$, that is, for both $p_{i}$ and
$p_{j}$, $t(q)=\frac{1}{2}$, but for the others, $t(q)=0 $. See the simplest
full-OACD generated from two entities in Fig.1.

\item[Step (3)] Repeat steps (1) and (2) for all $k=\frac{1}{2}n(n-1)$
point-pairs in $\boldsymbol{Q}$, and then overlay the $2k$ half-planes so that
they generate a spatial tessellation, containing a number of faces, edges, and vertexes.

\item[Step (4)] The chromatic code of each face, edge, and vertex is the sum
of the values $t(q)$ that are acquired from each half-plane partition, that
is,
\begin{equation}
\left(  p_{1}^{%
{\textstyle\sum\limits_{q\in\boldsymbol{Q}}}
t(q)},p_{2}^{%
{\textstyle\sum\limits_{q\in\boldsymbol{Q}}}
t(q)},\ldots,p_{i}^{%
{\textstyle\sum\limits_{q\in\boldsymbol{Q}}}
t(q)},\ldots,p_{n}^{%
{\textstyle\sum\limits_{q\in\boldsymbol{Q}}}
t(q)}\right)
\end{equation}

\end{description}

\begin{figure}
\centerline{\includegraphics[width=2in]{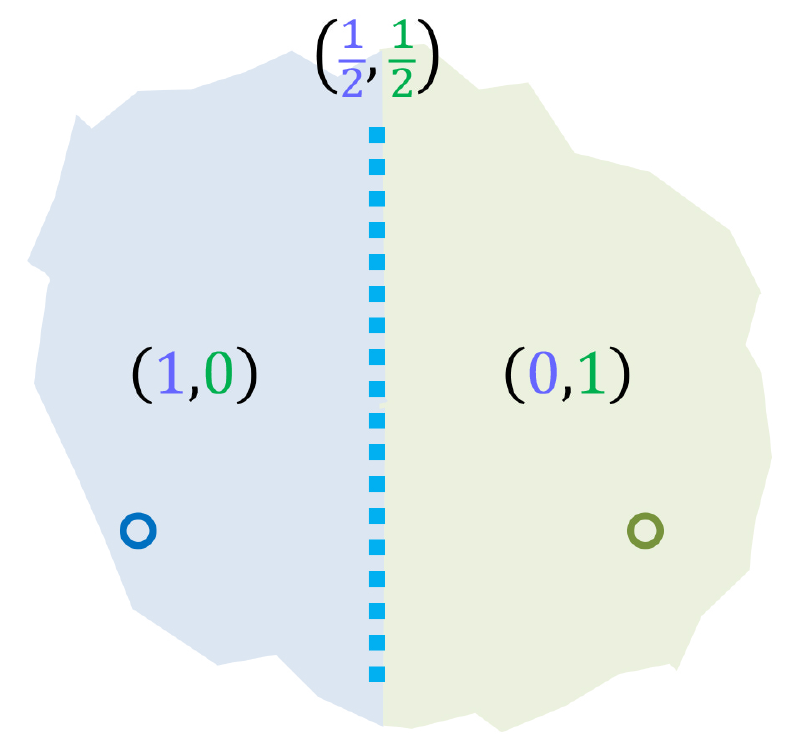}}
\caption{A full-coded OACD generated from 2 entities.}
\label{Fig1}
\end{figure}

Note that the point set $\boldsymbol{P}$ could be in any dimensional space $%
\mathbb{R}
^{m}$, and hence each partition divides the space into two half-spaces rather
than half-planes. This study mainly focuses on the planar full-OACDs in space
$%
\mathbb{R}
^{2}$. Fig.2 shows the procedure of generating a full-OACD from 3 points
(Fig.2a) in plane, denoted by $OACD(3,%
\mathbb{R}
^{2})$. Through half-plane partitions, we get 6 half-planes in Fig.2b-2d, then
we overlap them together into a diagram such that in Fig.2e, and finally we
sum the $t(q)$'s to compute chromatic code for each subspace in the diagram Fig.2f.

\begin{figure*}
\centerline{\includegraphics[width=5in]{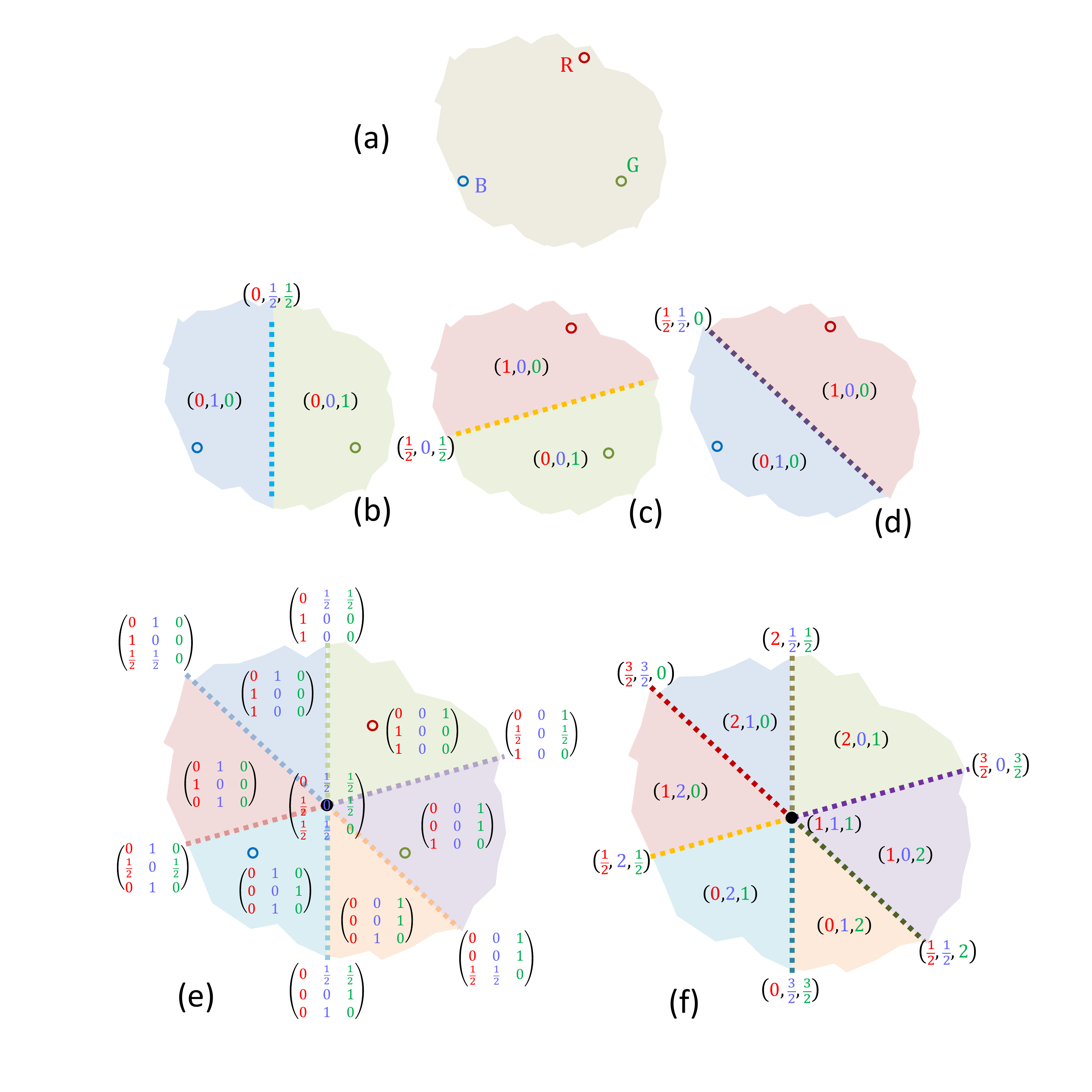}}
\caption{The procedure of generating a full-OACD(3, $\mathbb{R}^{2}$). (a) The generator set consists of three points marked with color R,
G, and B; (b)-(d) Half-plane partitions and assignments of chromatic codes
with respect to perpendicular bisectors $pb\left\langle B,G\right\rangle $, $%
pb\left\langle G,R\right\rangle $, and $pb\left\langle R,B\right\rangle $,
respectively. (e) Overlapping all the six half-planes in (b)-(d) together;
and (f) Adding all chromatic components together to form the chromatic codes.}
\label{Fig2}
\end{figure*}

In step (2), if we do not assign $t(q)=\frac{1}{2}$ to any bisectors, then the
obtained diagram is OACD. Therefore edges and vertexes in OACDs are without
codes. This makes the important difference bewteen OACD and full-OACD, where
edges and vertexes are with codes.

The subspaces, i.e., faces, edges, and vertexes generated in full-OACD are
called \emph{spatial particles }(denoted by $\Omega$), and faces are
particularly called \emph{cells} (denoted by $\zeta$), which has been
preliminarily studied in OACD \cite{Zhu2010}, \cite{Zhu2015}. Chromatic codes
of particles are $n$-tuples such as $\Omega(t_{1},t_{2},\ldots,t_{n}) $, in
which the number $t_{i}$ is called the \emph{chromatic component} of $p_{i}$
in the code, or the component at location $i$. Easy to know that $t_{i}$ will
be either integer or half-integer. Sometimes, if we are only interested in,
say, components of $p_{i}$ and $p_{j}$, then a chromatic code $\Omega
(t_{1},t_{2},\ldots,t_{i},\ldots t_{j},\ldots,t_{n})$ can be rewritten in a
short form such as $\Omega(t_{i},t_{j})\cup(T_{others})$, or just
$\Omega(t_{i},t_{j})$.

Fig.3 shows another two examples of full-OACDs. Fig.3a is an original
full-coded $OACD(4,%
\mathbb{R}
^{2})$ and Fig.3b is a homomorphic part of a full-coded $OACD(6,%
\mathbb{R}
^{2})$, where each spatial particles are coded in 6-tuples. Observing particle
patterns and codes in these full-OACDs we can find out many interesting properties.

\begin{definition}
Given a particle $\Omega(t_{1},t_{2},\ldots,t_{n})$, the ascending order of
its chromatic components is called the \emph{chromatic base} of the particle,
and denoted by $\beta(\Omega)=\{t_{1}^{\prime},t_{2}^{\prime},\ldots
t_{n}^{\prime}\}$.
\end{definition}

For example, cells $\zeta_{1}(0,2,3,1)$ and $\zeta_{2}(2,1,3,0)$ both have the
same base $\beta(\zeta_{1})=\beta(\zeta_{2})=\{0,1,2,3\}$. If two components
are equal, their orders are in random. For example, the base of edges
$(\frac{3}{2},0,3,\frac{3}{2})$ and $(\frac{3}{2},\frac{3}{2},3,0)$ are both
$\{0,\frac{3}{2},\frac{3}{2},3\}$. Chromatic codes are actually the
permutations of different bases. In previous studies, chromatic base was also
called the primary code of a cell \cite{Zhu2010}.

\begin{figure*}
\centerline{\includegraphics[width=5in]{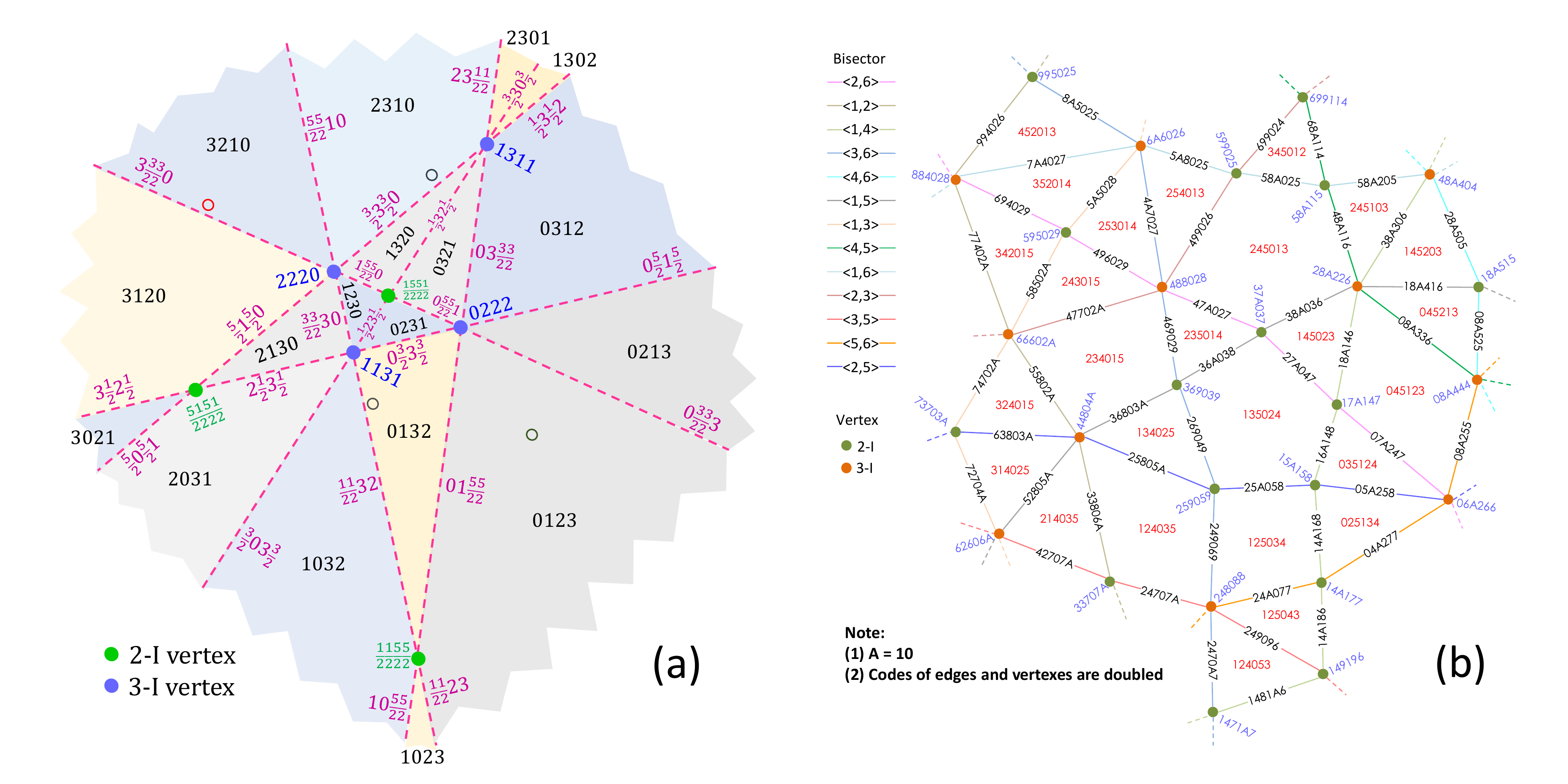}}
\caption{Two examples of full-OACDs. (a) a full-OACD(4, $%
\mathbb{R}
^{2}$); (b) Homomorphic part of a full-OACD(6, $%
\mathbb{R}
^{2}$).}
\label{Fig3}
\end{figure*}

\begin{definition}
If two particles $\Omega_{1}(t_{11}$, $t_{12}$, $\ldots$, $t_{1i}$, $\ldots$,
$t_{1n})$ and $\Omega_{2}(t_{21}$, $t_{22}$, $\ldots$, $t_{2i}$, $\ldots$,
$t_{2n})$ have the same chromatic codes, then they are called equi-color, and
denoted by $\Omega_{1}=\Omega_{2}$, that is,%
\begin{equation}
\forall i,t_{1i}=t_{2i}\Leftrightarrow\Omega_{1}=\Omega_{2}%
\end{equation}
otherwise, $\Omega_{1}\neq\Omega_{2}$.

If they have the same chromatic bases, then they are called equi-base, denoted
by $\Omega_{1}\cong\Omega_{2}$, that is, if $\beta(\Omega_{1})=\{t_{11}%
^{\prime}$, $t_{12}^{\prime}$, $\ldots$, $t_{1i}^{\prime}$, $\ldots$,
$t_{1n}^{\prime}\}$ and $\beta(\Omega_{2})=\{t_{21}^{\prime} $, $t_{22}%
^{\prime}$, $\ldots$, $t_{2i}^{\prime}$, $\ldots$, $t_{2n}^{\prime}\}$, then%
\begin{equation}
\forall i,t_{1i}^{\prime}=t_{2i}^{\prime}\Leftrightarrow\Omega_{1}\cong%
\Omega_{2}%
\end{equation}
otherwise, $\Omega_{1}\ncong\Omega_{2}$.
\end{definition}

\begin{property}
\label{P_base}Given two particles $\Omega_{1}$ and $\Omega_{2}$,%
\begin{equation}
\Omega_{1}=\Omega_{2}\Rightarrow\Omega_{1}\cong\Omega_{2}%
\end{equation}
and hence%
\begin{equation}
\Omega_{1}\ncong\Omega_{2}\Rightarrow\Omega_{1}\neq\Omega_{2}%
\end{equation}

\end{property}

This property indicates that if two cells are equi-color, they must be
equi-base, and if they are not equi-base, they are impossible to be the equi-color.

The number of cells, edges, and vertexes in a full-coded $OACD(n,%
\mathbb{R}
^{2})$ depends on the point pattern of the generator set $\boldsymbol{P}$.
This study mainly focuses on the \emph{general cases} of $\boldsymbol{P}$ in a
plane: (1) no more than two bisectors are parallel, and (2) no more than three
bisectors are concurrent, except that they are generated from the three
point-pairs which make a triangle.

\begin{definition}
In a general case of the point set $\boldsymbol{P}$, any three point-pairs
from three different points generate a vertex, called 3-I vertex (i.e., the
intersection of three perpendicular bisectors of a triangle), denoted by
$\varphi^{3I}$; and any two point-pairs from 4 different points generate a
vertex, called 2-I vertex (i.e., the intersection of two perpendicular
bisectors), denoted by $\varphi^{2I}$.
\end{definition}

Therefore vertexes $\varphi$ in full-coded $OACD(n,%
\mathbb{R}
^{2})$ are either 2-I or 3-I, see their examples in Fig.3.

\begin{property}
An $OACD(n,%
\mathbb{R}
^{2})$ contains $%
{\displaystyle\sum\nolimits_{i=1}^{C_{n}^{2}}}
i-C_{n}^{3}+1$ cells, $(C_{n}^{2})^{2}-3C_{n}^{3}$ edges, $C_{n}^{3}$ 3-I
vertexes, and $\frac{1}{2}C_{n}^{2}C_{n-2}^{2}$ 2-I vertexes.
\end{property}

\begin{proof}
The proof of the cell number could be referred to \cite{Zhu2010}. Here we only
prove the edge number. Suppose in a plane there are $n$ lines which intersect
with each other, then each line is divided into $n$ edges by the other $n-1$
lines, therefore the $n$ lines will generate $n^{2}$ edges. The total $n$
point will generate $C_{n}^{2}$ lines (bisectors) and hence $(C_{n}^{2})^{2}$
edges. But every three points generate a vertex which will reduce 3 edges,
therefore the total edge number will be $(C_{n}^{2})^{2}-3C_{n}^{3}$.
\end{proof}

\begin{property}
In an $OACD(n)$, the chromatic base of cells is%
\begin{equation}%
\mathbb{N}
=\{0,1,\ldots,n-1\}
\end{equation}

\end{property}

This property has been proved by \cite{Zhu2010}. It implies that all cells are
equi-base, and any two components of a cell are not equal. Below we use $%
\mathbb{N}
\lbrack i,j]$ to denote the integers between $i$ and $j$, and also including
$i$ and $j$.

\begin{property}
\label{P_Edge_Base}In an $OACD(n)$, the chromatic bases of edges are
\begin{equation}
\{%
\mathbb{N}
\backslash\{z,z+1\},z+\tfrac{1}{2},z+\tfrac{1}{2}\}
\end{equation}
for $z=%
\mathbb{N}
\lbrack0,n-2]$, meaning for each $z$ from $0$ to $n-2$, we obtain a base which
removes $z$ and $z+1$ from $N$ and then add two $z+\tfrac{1}{2}$.

Particularly, an edge (denote by $\eta$) generated by bisector $pb\left\langle
i,j\newline\right\rangle $ bears a code%
\begin{equation}
\eta(x_{i}^{z+\frac{1}{2}},x_{j}^{z+\frac{1}{2}})
\end{equation}
for $z=%
\mathbb{N}
\lbrack0,n-2]$.
\end{property}

\begin{proof}
Suppose $\eta$ is the edge between two cells $\zeta_{1}$ and $\zeta_{2}$,
therefore before the partition of $pb\left\langle i,j\newline\right\rangle $,
$\zeta_{1}$ and $\zeta_{2}$ should be merged into a larger cell $\zeta$ with
code $(x_{i}^{z},x_{j}^{z})$, that is, point $i$ and $j$ have the same
component $z$. After the partition, $\zeta_{1}$ and $\zeta_{2}$'s codes will
be $(x_{i}^{z+1},x_{j}^{z})$ and $(x_{i}^{z},x_{j}^{z+1})$, see the proof of
Lemma 2 in \cite{Zhu2010}. With respect to all other bisectors $pb\left\langle
i,x\newline\right\rangle $ or $pb\left\langle j\newline,x\right\rangle $,
$x\in\boldsymbol{I}\backslash\{i,j\}$, if $\zeta$ has not gained any
components, then minimum of $z$ could be 0; if $\zeta$ always gained one
component for all the other $n-2$ bisectors, then the maximum of $z$ could be
$n-2$. Therefore $\eta$'s chromatic code will be $(x_{i}^{z+\frac{1}{2}}%
,x_{j}^{z+\frac{1}{2}})$, and their bases will be $\{%
\mathbb{N}
\backslash\{z,z+1\},z+\tfrac{1}{2},z+\tfrac{1}{2}\}$, for $z=%
\mathbb{N}
\lbrack0,n-2]$.
\end{proof}

\begin{property}
\label{P_2I_Base}The chromatic bases of 2-I vertexes are
\begin{equation}
\{%
\mathbb{N}
\backslash\{z_{1},z_{2},z_{1}+1,z_{2}+1\},z_{1}+\tfrac{1}{2},z_{1}+\tfrac
{1}{2},z_{2}+\tfrac{1}{2},z_{2}+\tfrac{1}{2}\}
\end{equation}
for $z_{1}=%
\mathbb{N}
\lbrack0,n-4]$ and $z_{2}=%
\mathbb{N}
\lbrack z_{1}+2,n-2]$. Particularly, a vertex $\varphi^{2I}$ generated by two
bisectors $pb\left\langle i,j\newline\right\rangle $ and $pb\left\langle
u,v\right\rangle $ bears a code%
\begin{equation}
\varphi^{2I}(x_{i}^{z_{1}+\frac{1}{2}},x_{j}^{z_{1}+\frac{1}{2}},x_{u}%
^{z_{2}+\frac{1}{2}},x_{v}^{z_{2}+\frac{1}{2}})
\end{equation}
or%
\begin{equation}
\varphi^{2I}(x_{i}^{z_{2}+\frac{1}{2}},x_{j}^{z_{2}+\frac{1}{2}},x_{u}%
^{z_{1}+\frac{1}{2}},x_{v}^{z_{1}+\frac{1}{2}})
\end{equation}
for $z_{1}=%
\mathbb{N}
\lbrack0,n-4]$ and $z_{2}=%
\mathbb{N}
\lbrack z_{1}+2,n-2]$.
\end{property}

\begin{proof}
Suppose $pb\left\langle i,j\newline\right\rangle $ and $pb\left\langle
u,v\right\rangle $ are the last two bisectors partitioning a merged cell, then
according to the Lemma 2 in \cite{Zhu2010}, before the two partitions, the
cell should be with a code such as $(x_{i}^{z_{1}},x_{j}^{z_{1}},x_{u}^{z_{2}%
},x_{v}^{z_{2}})$. Let $z_{1}$ is the smaller integer, and then $z_{2}%
=z_{1}+\Delta$. After the two partitions by $pb\left\langle i,j\newline%
\right\rangle $ and $pb\left\langle u,v\right\rangle $, four new cells will be
generated with codes%
\begin{equation}
(x_{i}^{z_{1}},x_{j}^{z_{1}+1},x_{u}^{z_{1}+\Delta+1},x_{v}^{z_{1}+\Delta
})\label{eq_1}%
\end{equation}%
\begin{equation}
(x_{i}^{z_{1}+1},x_{j}^{z_{1}},x_{u}^{z_{1}+\Delta+1},x_{v}^{z_{1}+\Delta
})\label{eq_2}%
\end{equation}%
\begin{equation}
(x_{i}^{z_{1}},x_{j}^{z_{1}+1},x_{u}^{z_{1}+\Delta},x_{v}^{z_{1}+\Delta
+1})\label{eq_3}%
\end{equation}%
\begin{equation}
(x_{i}^{z_{1}+1},x_{j}^{z_{1}},x_{u}^{z_{1}+\Delta},x_{v}^{z_{1}+\Delta
+1})\label{eq_4}%
\end{equation}
If $\Delta=0$ or $1$, then we can always find that in some codes of
Eq.(\ref{eq_1})-(\ref{eq_4}), two components are equal. For example, if
$\Delta=0$, there are two $z_{1}$'s and two $z_{1}+1$'s in Eq.(\ref{eq_1}),
and if $\Delta=1$, there are two $z_{1}+1$'s in Eq.(\ref{eq_2}). But cellular
base is $%
\mathbb{N}
$, meaning any two components are not equal, therefore $\Delta\geq2$. Because
$pb\left\langle i,j\newline\right\rangle $ and $pb\left\langle
u,v\right\rangle $ involve 4 points, then the maximum of $z_{1}$ should be
$n-4$, and hence $z_{1}=%
\mathbb{N}
\lbrack0,n-4]$, $z_{2}=%
\mathbb{N}
\lbrack z_{1}+2,n-2]$. The remainder of the proof follows along the line of
the proof of Property \ref{P_Edge_Base}.
\end{proof}

\begin{property}
\label{P_3I_Base}The chromatic bases of 3-I vertexes are%
\begin{equation}
\{%
\mathbb{N}
\backslash\{z,z+1,z+2\},z+1,z+1,z+1\}\label{Eq_3I_base}%
\end{equation}
for $z=%
\mathbb{N}
\lbrack0,n-3]$. Particularly, a vertex $\varphi^{3I}$ generated by three
bisectors $pb\left\langle i,j\newline\right\rangle $, $pb\left\langle
j\newline,k\right\rangle $, and $pb\left\langle k,i\right\rangle $ bears a
code%
\begin{equation}
\varphi^{3I}(x_{i}^{z+1},x_{j}^{z+1},x_{k}^{z+1})
\end{equation}
for $z=%
\mathbb{N}
\lbrack0,n-3]$.
\end{property}

\begin{proof}
Suppose before the partitions of $pb\left\langle i,j\newline\right\rangle $,
$pb\left\langle j\newline,k\right\rangle $, and $pb\left\langle
k,i\right\rangle $, the merged cell has a code $(x_{i}^{z},x_{j}^{z+\Delta
_{1}},x_{k}^{z+\Delta_{2}})$, where $\Delta_{1}\geq0$ and $\Delta_{2}\geq0$.
After the partitions, six new cells will be generated with codes%
\begin{equation}
(x_{i}^{z+2},x_{j}^{z+\Delta_{1}+1},x_{k}^{z+\Delta_{2}})\cup(X_{others}%
)\label{Eq_3IBase_1}%
\end{equation}%
\begin{equation}
(x_{i}^{z+2},x_{j}^{z+\Delta_{1}},x_{k}^{z+\Delta_{2}+1})\cup(X_{others}%
)\label{Eq_3IBase_2}%
\end{equation}%
\begin{equation}
(x_{i}^{z+1},x_{j}^{z+\Delta_{1}},x_{k}^{z+\Delta_{2}+2})\cup(X_{others}%
)\label{Eq_3IBase_3}%
\end{equation}%
\begin{equation}
(x_{i}^{z+1},x_{j}^{z+\Delta_{1}+2},x_{k}^{z+\Delta_{2}})\cup(X_{others}%
)\label{Eq_3IBase_4}%
\end{equation}%
\begin{equation}
(x_{i}^{z},x_{j}^{z+\Delta_{1}+1},x_{k}^{z+\Delta_{2}+2})\cup(X_{others}%
)\label{Eq_3IBase_5}%
\end{equation}%
\begin{equation}
(x_{i}^{z},x_{j}^{z+\Delta_{1}+2},x_{k}^{z+\Delta_{2}+1})\cup(X_{others}%
)\label{Eq_3IBase_6}%
\end{equation}
We examine the below possible values of $\Delta_{1}$ and $\Delta_{2}$.

(1) $\Delta_{1}=1$ or $\Delta_{1}=2$, $\Delta_{2}=1$ or $\Delta_{2}=2$.

If $\Delta_{1}=1$ or $\Delta_{1}=2$, for example, in Eq.(\ref{Eq_3IBase_3})
and (\ref{Eq_3IBase_2}) there will be two components equalling $z+1$ or $z+2$;
similarly, if $\Delta_{2}=1$ or $\Delta_{2}=2$, in Eq.(\ref{Eq_3IBase_4}) and
(\ref{Eq_3IBase_1}) there will be two components equalling $z+1$ or $z+2$.

(2) $\Delta_{1}\geq3$, $\Delta_{2}\geq3$.

According to Eq.(\ref{Eq_3IBase_5}) and (\ref{Eq_3IBase_6}), there must be
values $z+1$ and $z+2$ in $X_{others}$, because they are not in locations
$x_{i}$, $x_{j}$, or $x_{k}$. However, according to Eq.(\ref{Eq_3IBase_1}%
)-(\ref{Eq_3IBase_4}), $z+1$ and $z+2$ are already in $x_{i}$, so that they
cannot be in $X_{others}$.

From the above two cases we know that the only allowed values of $\Delta_{1} $
and $\Delta_{2}$ are both 0, and the merged cell must bear a code%
\begin{equation}
(x_{i}^{z},x_{j}^{z},x_{k}^{z})
\end{equation}

Then at the intersection of the three bisectors, the 3-I vertex acquires
components $\tfrac{1}{2}$ at $x_{i}$ and $\tfrac{1}{2}$ at $x_{j}$ from
$pb\left\langle i,j\newline\right\rangle $, $\tfrac{1}{2}$ at $x_{j}$ and
$\tfrac{1}{2}$ at $x_{k}$ from $pb\left\langle j,k\newline\right\rangle $,
$\tfrac{1}{2}$ at $x_{k}$ and $\tfrac{1}{2}$ at $x_{i}$ from $pb\left\langle
k,i\newline\right\rangle $, and therefore gain a code%
\begin{equation}
(x_{i}^{z+\frac{1}{2}+\frac{1}{2}},x_{j}^{z+\frac{1}{2}+\frac{1}{2}}%
,x_{k}^{z+\frac{1}{2}+\frac{1}{2}})=(x_{i}^{z+1},x_{j}^{z+1},x_{k}^{z+1})
\end{equation}
From Eq.(\ref{Eq_3IBase_1})-(\ref{Eq_3IBase_6}), we know that $X_{others}$ do
not contain components $z$, $z+1$, and $z+2$, thus we know the base of the 3-I
vertex is in form of Eq.(\ref{Eq_3I_base}). Because the range of $z$ in a cell
is from 0 to $n-1$, the minimum $z$ should be 0 and maximum $z$ should be
$z+2=n-1\Rightarrow$ $z=n-3$.
\end{proof}

This property indicates that the chromatic codes of 3-I vertexes contain three
identical integers which are different from the rest integers in codes. If
cancel one $z+1$, Eq.(\ref{Eq_3I_base}) can be rewritten as
\begin{equation}
\{%
\mathbb{N}
\backslash\{z,z+2\},z+1,z+1\}
\end{equation}
for $z=%
\mathbb{N}
\lbrack0,n-3]$.

\begin{theorem}
\label{Theorem_unique_base}Different types of particles in a full-OACD are not
equi-base, that is,%
\begin{equation}
\zeta\ncong\eta\ncong\varphi^{_{2I}}\ncong\varphi^{_{3I}}%
\end{equation}

\end{theorem}

This theorem provides an approach to determine particle types. For example, if
we see a particle with chromatic components being all different integers, then
it must be a cell; if it contains 2 half-integers, it must be an edge; if it
contains 3 equal integers, it must be a 3-I vertex; and if contains 4
half-integers, it must be a 2-I vertex.

\begin{notation}
The component-counting function $H(\Omega,m)$ is a function counting the
number of $m$ in the chromatic code of $\Omega$, that is, the function tells
how many components equal to $m$.
\end{notation}

\begin{definition}
The difference tuple of two particles $\Omega_{1}(t_{11}$, $t_{12}$, $\ldots$,
$t_{1n})$ and $\Omega_{2}(t_{21}$, $t_{22}$, $\ldots$, $t_{2n})$ is defined
by
\begin{align}
\Psi(\Omega_{1},\Omega_{2})  & =(\psi_{1},\psi_{2},\ldots,\psi_{n})\\
& =(|t_{11}-t_{21}|,|t_{12}-t_{22}|,\ldots,|t_{1n}-t_{2n}|)\nonumber
\end{align}
where $\psi_{i}=|t_{1i}-t_{2i}|$.

Then the chromatic distance between the two particles is defined by
\begin{equation}
\delta(\Omega_{1},\Omega_{2})=\sum_{i=1}^{n}\psi_{i}%
\end{equation}
and each $\psi_{i}$ is called the chromatic distance at the component $i$, and
denoted by $\delta(\psi_{i})$.

In addition, the code distance between two particles is defined by%
\begin{equation}
\gamma(\Omega_{1},\Omega_{2})=n-H(\delta(\Omega_{1},\Omega_{2}),0)
\end{equation}

\end{definition}

The chromatic distance is also called transition number $T$ between two cells
in our previous study, and it is actually the Manhattan distance between two
particles. The code distance is actually the Hamming distance between two
particles if we treat their codes and components as strings rather than numbers.

\begin{definition}
The union of $m$ particles $\Omega_{1}(t_{11}$, $t_{12}$, $\ldots$, $t_{1n})$,
$\Omega_{2}(t_{21}$, $t_{22}$, $\ldots$, $t_{2n})$, \ldots, $\Omega_{m}%
(t_{m1}$, $t_{m2}$, $\ldots$, $t_{mn})$ is called a complex or a m-complex,
denoted by $\Theta$, and its code is given by%
\begin{align}
\Theta\{\Omega_{1},\Omega_{2},\ldots,\Omega_{m}\}  & =\sum_{i=1}^{m}\Omega
_{m}\\
& =\left(  \sum_{i=1}^{m}t_{i1},\sum_{i=1}^{m}t_{i2},\ldots,\sum_{i=1}%
^{m}t_{in}\right)  \text{.}\nonumber
\end{align}
These $m$ particles are called the elemental particles of the m-complex. If
the $m$ particles are all cells, then the m-complex is also called a m-cell
cluster. One particle could be taken as a 1-complex.
\end{definition}

\begin{theorem}
\label{Theorem_unique_cell}If $\zeta_{1}$ and $\zeta_{2}$ are two cells in a
full-OACD, then
\begin{equation}
\zeta_{1}\neq\zeta_{2},\label{Eq_unique_cell}%
\end{equation}

\end{theorem}

This theorem has been proved by \cite{Zhu2010}. It tells that any two cells
are not equi-color -- their codes are unique.

Because any vertex in OACD is either 2-I or 3-I, therefore the a full-OACD is
tessellated by two types of structural units such as the two in Fig.4: the one
containing $\varphi^{2I}$ is called \emph{2-I unit }(Fig.4a), and the other
containing $\varphi^{3I}$ is called \emph{3-I unit }(Fig.4b). According to the
proofs of Property \ref{P_2I_Base} and \ref{P_3I_Base}, particle codes in
2-I/3-I units should be those shown in Fig.4, and then it is easy to calculate
and prove the below four properties.

\begin{figure*}
\centerline{\includegraphics[width=5in]{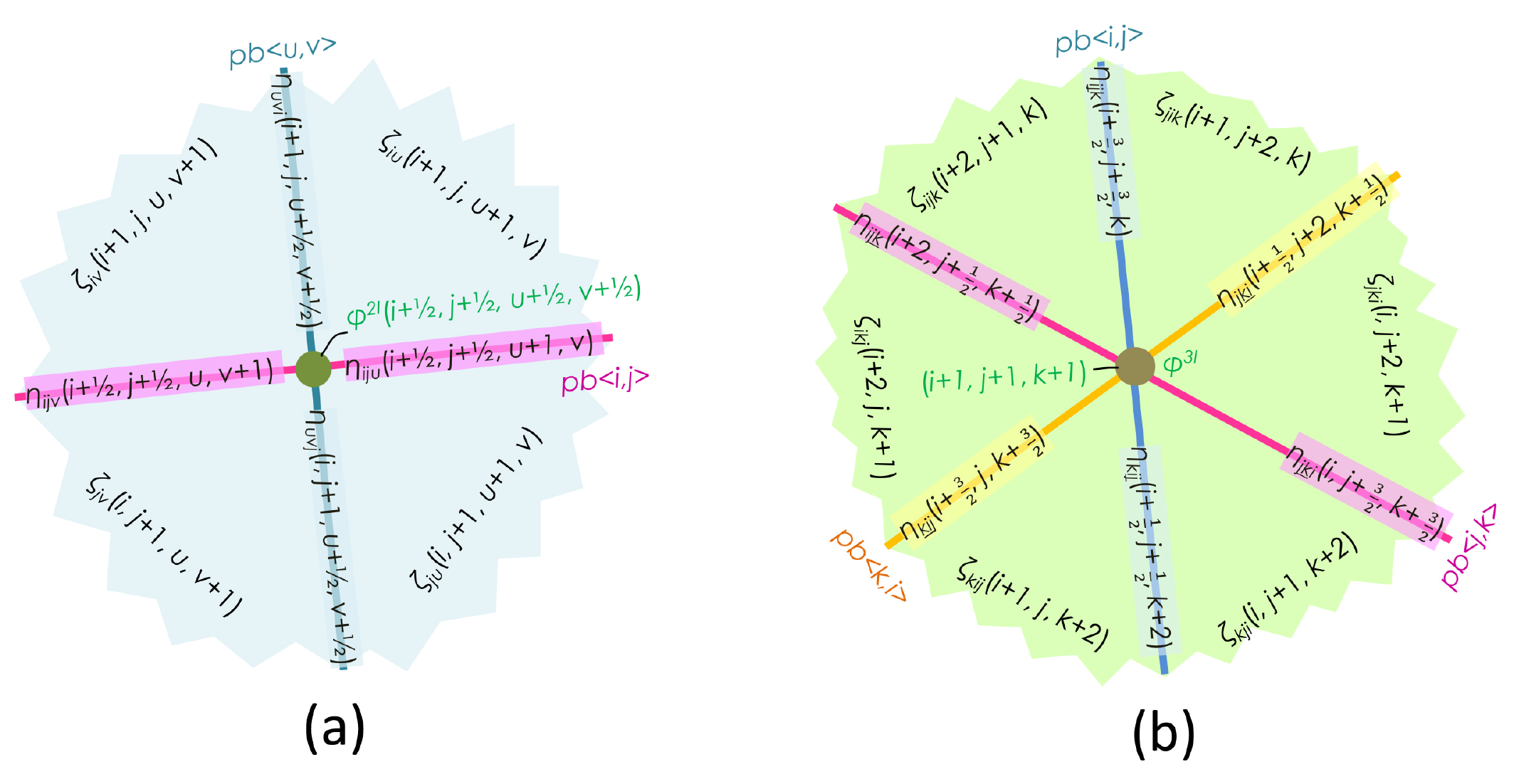}}
\caption{Two basic structural units of full-OACD. (a) 2-I unit; (b)
3-I unit.}
\label{Fig4}
\end{figure*}

\begin{figure*}
\centerline{\includegraphics[width=4.5in]{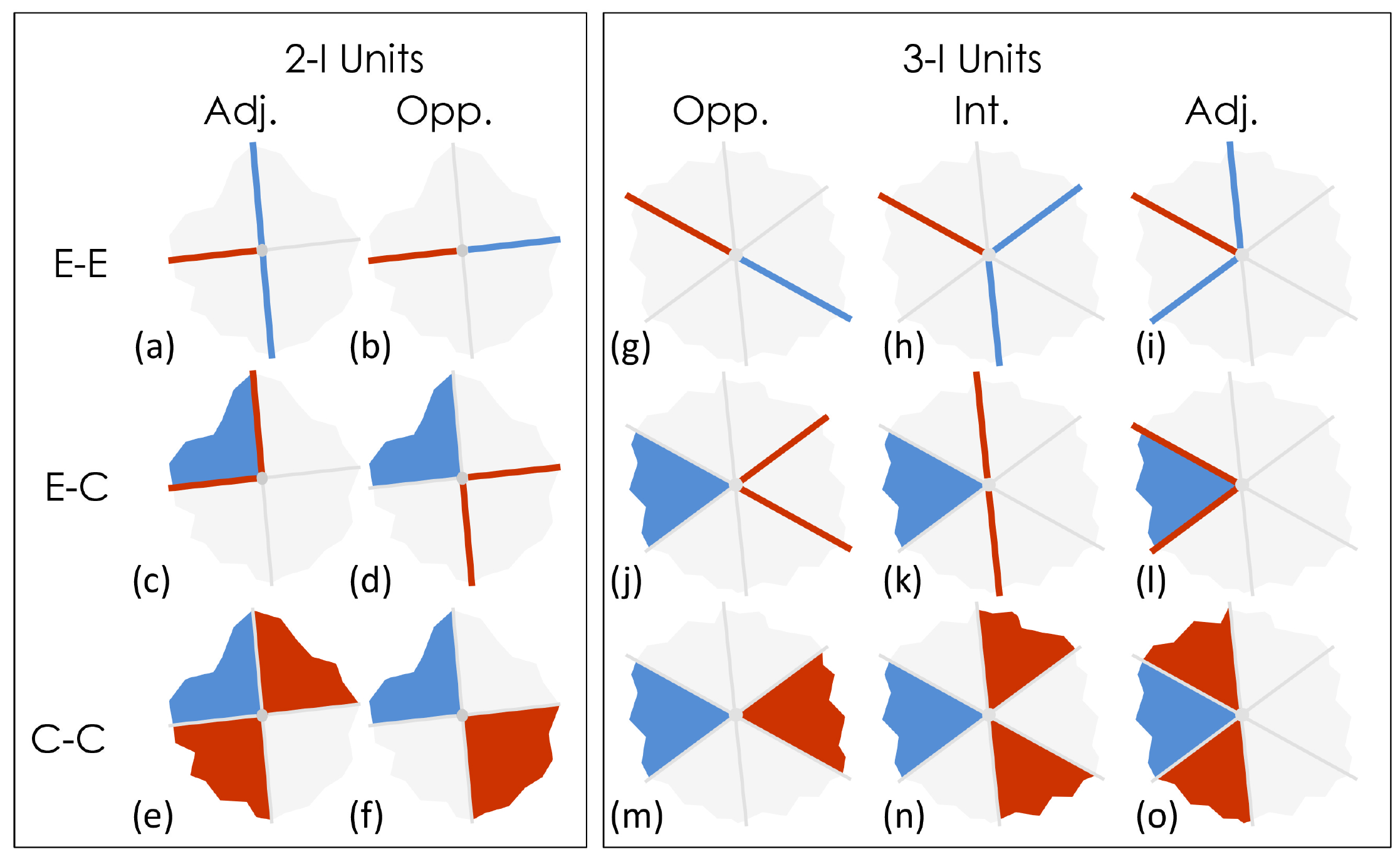}}
\caption{Three types of particle relations in 2-I/3-I units:
adjacent (Adj.), interval (Int.) and opposite (Opp.).}
\label{Fig5}
\end{figure*}

\begin{property}
\label{P_2I_Codes}A 2-I unit generated by $pb\left\langle i,j\right\rangle $
and $pb\left\langle u,v\right\rangle $ contains the following 9
particles.\newline

(1) One 2-I vertex with code%
\begin{equation}
\varphi_{ijuv}^{2I}(x_{i}+\tfrac{1}{2},x_{j}+\tfrac{1}{2},x_{u}+\tfrac{1}%
{2},x_{v}+\tfrac{1}{2})
\end{equation}

(2) Four edges with codes%
\begin{align}
& \eta_{iju}(x_{i}+\tfrac{1}{2},x_{j}+\tfrac{1}{2},x_{u}+1,x_{v}),\eta
_{ijv}(x_{i}+\tfrac{1}{2},x_{j}+\tfrac{1}{2},x_{u},x_{v}+1)\\
& \eta_{uvi}(x_{i}+1,x_{j},x_{u}+\tfrac{1}{2},x_{v}+\tfrac{1}{2}),\eta
_{uvj}(x_{i},x_{j}+1,x_{u}+\tfrac{1}{2},x_{v}+\tfrac{1}{2})\nonumber
\end{align}

(3) Four cells with codes%
\begin{align}
& \zeta_{iu}(x_{i}+1,x_{j},x_{u}+1,x_{v}),\zeta_{iv}(x_{i}+1,x_{j},x_{u}%
,x_{v}+1)\\
& \zeta_{ju}(x_{i},x_{j}+1,x_{u}+1,x_{v}),\zeta_{jv}(x_{i},x_{j}+1,x_{u}%
,x_{v}+1)\nonumber
\end{align}

\end{property}

\begin{property}
\label{P_3I_Codes}A 3-I unit generated by $pb\left\langle i,j\right\rangle $,
$pb\left\langle j,k\right\rangle $ and $pb\left\langle k,i\right\rangle $
contains the following 13 particles.

(1) One 3-I vertex with code%
\begin{equation}
\varphi_{ijk}^{3I}(x_{i}+1,x_{j}+1,x_{k}+1)
\end{equation}

(2) Six edges with codes%
\begin{align}
& \eta_{\underline{jk}i}(x_{i},x_{j}+\tfrac{3}{2},x_{k}+\tfrac{3}{2}%
),\eta_{i\underline{jk}}(x_{i}+2,x_{j}+\tfrac{1}{2},x_{k}+\tfrac{1}%
{2})\label{Eq_3ICodes_edges}\\
& \eta_{\underline{ki}j}(x_{i}+\tfrac{3}{2},x_{j},x_{k}+\tfrac{3}{2}%
),\eta_{j\underline{ki}}(x_{i}+\tfrac{1}{2},x_{j}+2,x_{k}+\tfrac{1}%
{2})\nonumber\\
& \eta_{\underline{ij}k}(x_{i}+\tfrac{3}{2},x_{j}+\tfrac{3}{2},x_{k}%
),\eta_{k\underline{ij}}(x_{i}+\tfrac{1}{2},x_{j}+\tfrac{1}{2},x_{k}%
+2)\nonumber
\end{align}
Note, in Eq.\ref{Eq_3ICodes_edges}, the underlined index indicates the
perpendicular bisector which makes the edge.

(3) Six cells with codes%
\begin{align}
& \zeta_{ijk}(x_{i}+2,x_{j}+1,x_{k}),\zeta_{ikj}(x_{i}+2,x_{j},x_{k}%
+1)\label{Eq_3ICodes_cells}\\
& \zeta_{jik}(x_{i}+1,x_{j}+2,x_{k}),\zeta_{jki}(x_{i},x_{j}+2,x_{k}%
+1)\nonumber\\
& \zeta_{kij}(x_{i}+1,x_{j},x_{k}+2),\zeta_{kji}(x_{i},x_{j}+1,x_{k}%
+2)\nonumber
\end{align}

\end{property}

\begin{property}
\label{P_2I_Unit}In 2-I unit space:

(1) The codes of the vertex $\varphi^{2I}$ is the \emph{average} of (I) two
edges which are in the same bisectors, (II) the two cells which are opposite
to the vertex, (III) all the four edges, and (IV) all the four cells, that is,%
\begin{align}
\varphi^{2I}  & =\tfrac{1}{2}(\eta_{iju}+\eta_{ijv})=\tfrac{1}{2}(\eta
_{uvi}+\eta_{uvj})\\
& =\tfrac{1}{2}(\zeta_{iu}+\zeta_{jv})=\tfrac{1}{2}(\zeta_{iv}+\zeta
_{jv})\nonumber\\
& =\tfrac{1}{4}(\eta_{iju}+\eta_{ijv}+\eta_{uvi}+\eta_{uvj})\nonumber\\
& =\tfrac{1}{4}(\zeta_{iu}+\zeta_{iv}+\zeta_{ju}+\zeta_{jv})\newline\nonumber
\end{align}

(2) The codes of an edge $\eta$ is the half of the two cells $\zeta_{1}$ and
$\zeta_{2}$ which are respectively on the two sides of the edge. If
$\xi=\{\zeta_{1},\zeta_{2}\}$, then,%
\begin{equation}
\eta=\tfrac{1}{2}(\zeta_{1}+\zeta_{2})=\tfrac{1}{2}\xi
\label{Eq_edge_half_cluster}%
\end{equation}

(3) The two edges are equi-base if they in the same bisector, but not
equi-base if they are not in the same bisector, that is,%
\begin{gather}
\eta_{iju}\cong\eta_{ijv},\eta_{uvi}\cong\eta_{uvj}\\
\eta_{iju}\ncong\eta_{uvi},\eta_{ijv}\ncong\eta_{uvj}\nonumber
\end{gather}

\end{property}

\begin{property}
\label{P_3I_Units}In 3-I unit space:\newline

(1) The codes of the vertex $\varphi^{3I}$ is the average of (I) all the six
edges/cells, (II) the three edges/cells which are interval with each other,
and (III) the two edges/cells which are opposite to each other (for edges in
this case, they are in the same bisector), that is,%
\begin{align}
\varphi^{3I}  & =\tfrac{1}{2}(\eta_{\underline{ki}j}+\eta_{j\underline{ki}%
})=\tfrac{1}{2}(\eta_{\underline{ij}k}+\eta_{k\underline{ij}})=\tfrac{1}%
{2}(\eta_{\underline{jk}i}+\eta_{i\underline{jk}})\\
& =\tfrac{1}{2}(\zeta_{kij}+\zeta_{jik})=\tfrac{1}{2}(\zeta_{ijk}+\zeta
_{kji})=\tfrac{1}{2}(\zeta_{ikj}+\zeta_{jki})\nonumber\\
& =\tfrac{1}{3}(\eta_{\underline{ki}j}+\eta_{\underline{ij}k}+\eta
_{\underline{jk}i})=\tfrac{1}{3}(\eta_{j\underline{ki}}+\eta_{k\underline{ij}%
}+\eta_{i\underline{jk}})\nonumber\\
& =\tfrac{1}{3}(\zeta_{ijk}+\zeta_{kij}+\zeta_{jki})=\tfrac{1}{3}(\zeta
_{jik}+\zeta_{ikj}+\zeta_{kji})\nonumber\\
& =\tfrac{1}{6}(\eta_{\underline{ki}j}+\eta_{\underline{ij}k}+\eta
_{\underline{jk}i}+\eta_{j\underline{ki}}+\eta_{k\underline{ij}}%
+\eta_{i\underline{jk}})\nonumber\\
& =\tfrac{1}{6}(\zeta_{ijk}+\zeta_{kij}+\zeta_{jki}+\zeta_{jik}+\zeta
_{ikj}+\zeta_{kji})\newline\nonumber
\end{align}

(2) The edge codes have the same property as the Property \ref{P_2I_Unit}.(2).

(3) The two edges in the same bisectors are not equi-base, but the three
interval edges are equi-base, that is,%
\begin{gather}
\eta_{\underline{ki}j}\ncong\eta_{j\underline{ki}},\eta_{\underline{ij}%
k}\ncong\eta_{k\underline{ij}},\eta_{\underline{jk}i}\ncong\eta
_{i\underline{jk}}\\
\eta_{\underline{ki}j}\cong\eta_{\underline{ij}k}\cong\eta_{\underline{jk}%
i},\eta_{j\underline{ki}}\cong\eta_{k\underline{ij}}\cong\eta_{i\underline{jk}%
}\nonumber
\end{gather}

\end{property}

The spatial relations among particles in 2-I/3-I units have three types:
adjacent, interval, and opposite, see Fig.5. If spatial relations between two
particles in 2-I/3-I units are different, their chromatic and code distances
are also different, see the below Property \ref{P_2I3IUnits_Distance}.

\begin{property}
\label{P_2I3IUnits_Distance}Within a 2-I or 3-I unit of a full-coded $OACD(n,%
\mathbb{R}
^{2})$, the chromatic distance $\delta$ and code distance $\gamma$ between two
particles $\Omega_{1}$ and $\Omega_{2}$ are listed in Table \ref{Table1}.%

\begin{table}[tbp] \centering
\begin{tabular}
[c]{ccccccccc}\hline\hline
$\Omega_{\mathbf{1}}$ & $\Omega_{\mathbf{2}}$ & Relation &
\multicolumn{3}{c}{2-I Units} & \multicolumn{3}{c}{3-I Units}\\
&  &  & $\delta$ & $\gamma$ & Base & $\delta$ & $\gamma$ & Base\\\hline
\multicolumn{1}{l}{Vertex $\varphi$} & \multicolumn{1}{l}{Edge $\eta$} & - &
1 & 2 & $\ncong$ & 2 & 3 & $\ncong$\\
\multicolumn{1}{l}{Vertex $\varphi$} & \multicolumn{1}{l}{Cell $\zeta$} & - &
2 & 4 & $\ncong$ & 2 & 2 & $\ncong$\\
\multicolumn{1}{l}{Edge $\eta_{1}$} & \multicolumn{1}{l}{Edge $\eta_{2}$} &
\multicolumn{1}{l}{Adjacent} & 2 & 4 & $\ncong$ & 2 & 3 & $\ncong$\\
\multicolumn{1}{l}{Edge $\eta_{1}$} & \multicolumn{1}{l}{Edge $\eta_{2}$} &
\multicolumn{1}{l}{Opposite} & 2 & 2 & $\cong$ & 4 & 3 & $\ncong$\\
\multicolumn{1}{l}{Edge $\eta_{1}$} & \multicolumn{1}{l}{Edge $\eta_{2}$} &
\multicolumn{1}{l}{Interval} & - & - & - & 3 & 2 & $\cong$\\
\multicolumn{1}{l}{Edge $\eta$} & \multicolumn{1}{l}{Cell $\zeta$} &
\multicolumn{1}{l}{Adjacent} & 1 & 2 & $\ncong$ & 1 & 2 & $\ncong$\\
\multicolumn{1}{l}{Edge $\eta$} & \multicolumn{1}{l}{Cell $\zeta$} &
\multicolumn{1}{l}{Opposite} & 3 & 4 & $\ncong$ & 4 & 3 & $\ncong$\\
\multicolumn{1}{l}{Edge $\eta$} & \multicolumn{1}{l}{Cell $\zeta$} &
\multicolumn{1}{l}{Interval} & - & - & - & 3 & 3 & $\ncong$\\
\multicolumn{1}{l}{Cell $\zeta_{1}$} & \multicolumn{1}{l}{Cell $\zeta_{2}$} &
\multicolumn{1}{l}{Adjacent} & 2 & 2 & $\cong$ & 2 & 2 & $\cong$\\
\multicolumn{1}{l}{Cell $\zeta_{1}$} & \multicolumn{1}{l}{Cell $\zeta_{2}$} &
\multicolumn{1}{l}{Opposite} & 4 & 4 & $\cong$ & 4 & 2 & $\cong$\\
\multicolumn{1}{l}{Cell $\zeta_{1}$} & \multicolumn{1}{l}{Cell $\zeta_{2}$} &
\multicolumn{1}{l}{Interval} & - & - & - & 4 & 3 & $\cong$\\\hline
\end{tabular}
\caption{Chromatic and code distances between two particles in 2-I/3-I unit
space. Their spatial relations are shown in Fig.5}\label{Table1}%
\end{table}%

\end{property}

An important requirement for full-OACDs is that we expect their particle codes
to be unique.

\begin{theorem}
\label{Theorem_unique_edge}If $\eta_{1}$ and $\eta_{2}$ are two edges in a
full-OACD, then%
\begin{equation}
\eta_{1}\neq\eta_{2}.\label{Eq_unique_edge}%
\end{equation}

\end{theorem}

\begin{proof}
Suppose $\zeta_{1Left}$ and $\zeta_{1Right}$ are two cells beside $\eta_{1}$,
and $\zeta_{2Left}$ and $\zeta_{2Right}$ are two cells beside $\eta_{2}$, and
then they make two 2-cell clusters $\xi_{1}(\zeta_{1Left},\zeta_{1Right})$ and
$\xi_{2}(\zeta_{2Left},\zeta_{2Right})$, respectively. It has been proved that
chromatic codes of any connected 2-cell cluster $\xi$ are unique
\cite{Zhu2015}, that is, $\xi_{1}\neq\xi_{2}$. Then according to
Eq.\ref{Eq_edge_half_cluster}, $\eta_{1}=\frac{1}{2}\xi_{1}\neq\frac{1}{2}%
\xi_{2}=\eta_{2}$.
\end{proof}

This theorem tells that chromatic codes of edges are unique.

\begin{theorem}
\label{Theorem_unique_vertex}If $\varphi_{1}$ and $\varphi_{2}$ are two
vertexes in a full-OACD, then%
\begin{equation}
\varphi_{1}\neq\varphi_{2}.\label{Eq_unique_vertex}%
\end{equation}

\end{theorem}

\begin{proof}
Case (1): One vertex is 2-I and the other vertex is 3-I.

According to Property \ref{P_base}, $\varphi^{_{2I}}\ncong\varphi^{_{3I}}$, so
$\varphi_{1}\neq\varphi_{2}$.

Case (2): They are both 3-I vertexes.

Suppose $\varphi_{1}^{3I}$ and $\varphi_{2}^{3I}$ are different 3-I vertexes
but with the same code $(x_{i}^{z},x_{j}^{z},x_{k}^{z})$, i.e., they have
three chromatic components which are the same integer $z$ given by three
points $i$, $j$, and $k$. However, the bisectors generated from 3 point can
only intersect at one 3-I vertex, so if $\varphi_{1}$ and $\varphi_{2}$ are
different vertexes, their codes are impossible to be the same, i.e.,
$\varphi_{1}^{3I}\neq\varphi_{2}^{3I}$.

(3) They are both 2-I vertexes.

Suppose $\varphi_{1}^{2I}$ and $\varphi_{2}^{2I}$ are two different 2-I
vertexes. The $\varphi_{1}^{2I}$ was generated by $pb\left\langle i_{1}%
,j_{1}\right\rangle $ and $pb\left\langle u_{1},v_{1}\right\rangle $, and
hence with a code $(x_{i_{1}}^{z_{1}},x_{j_{1}}^{z_{1}},x_{u_{1}}^{z_{2}%
},x_{v_{1}}^{z_{2}})$; The $\varphi_{2}^{2I}$ was generated by $pb\left\langle
i_{2},j_{2}\right\rangle $ and $pb\left\langle u_{2},v_{2}\right\rangle $, and
hence with a code $(x_{i_{2}}^{z_{1}},x_{j_{2}}^{z_{1}},x_{u_{2}}^{z_{2}%
},x_{v_{2}}^{z_{2}})$. The only way to make $\varphi_{1}^{2I}=\varphi_{2}%
^{2I}$ is that $i_{1}=i_{2}$, $j_{1}=j_{2}$, $u_{1}=u_{2}$, and $v_{1}=v_{2}$,
but this makes $\varphi_{1}^{2I}$ and $\varphi_{2}^{2I}$ are the same vertex,
because two bisectors can only intersect at one 2-I vertex. Therefore if
$\varphi_{1}^{2I}$ and $\varphi_{2}^{2I}$ are two different 2-I vertexes,
their codes are impossible to be equal.

Based on the above three cases, we know that $\varphi_{1}$ $\neq\varphi_{2}$
\end{proof}

This theorem indicates that chromatic codes of vertexes are also unique.
Ultimately, according to Theorems \ref{Theorem_unique_base}%
-\ref{Theorem_unique_vertex}, we obtain the below corollary.

\begin{corollary}
Chromatic particle codes in a full-OACD are unique, that is, given two
particles $\Omega_{1},\Omega_{2}$ $\in OACD(n,%
\mathbb{R}
^{2})$,%
\begin{equation}
\Omega_{1}\neq\Omega_{2}\label{Eq_unique_codes}%
\end{equation}

\end{corollary}

\section{Spatial particle topology in full-OACD}

A planar full-OACD contains three types of particles: vertexes, edges, and
cells. Spatial topological relations among these particles are usually similar
to those conventional relations for vectorial geometry in GIS, such as equal,
adjacent, disjoint, and overlap. Fig.6 shows the spatial relations between two
particles investigated in this study, and these relations can be simply
represented and calculated by using chromatic codes. In addition, the more
complicated spatial relations among $m$-complexes can be also reasoned from
analyzing their chromatic codes. Below we demonstrate the major spatial
topology among particles and complexes, as well as the relations between them
and chromatic codes, in which we particularly focus on cells and clusters.
Property \ref{P_2I3IUnits_Distance} already gives the conditions from which we
know that different particle relations bear different chromatic and code
distances, however, those are only necessary conditions. In this section we
will give proofs that those conditions are also sufficient so that we can use
two distances $\delta(\Omega_{1},\Omega_{2}) $ and $\gamma(\Omega_{1}%
,\Omega_{2})$, and their bases $\beta(\Omega_{1})$ and $\beta(\Omega_{2})$ to
determine their topological spatial relations.

\begin{figure*}
\centerline{\includegraphics[width=4in]{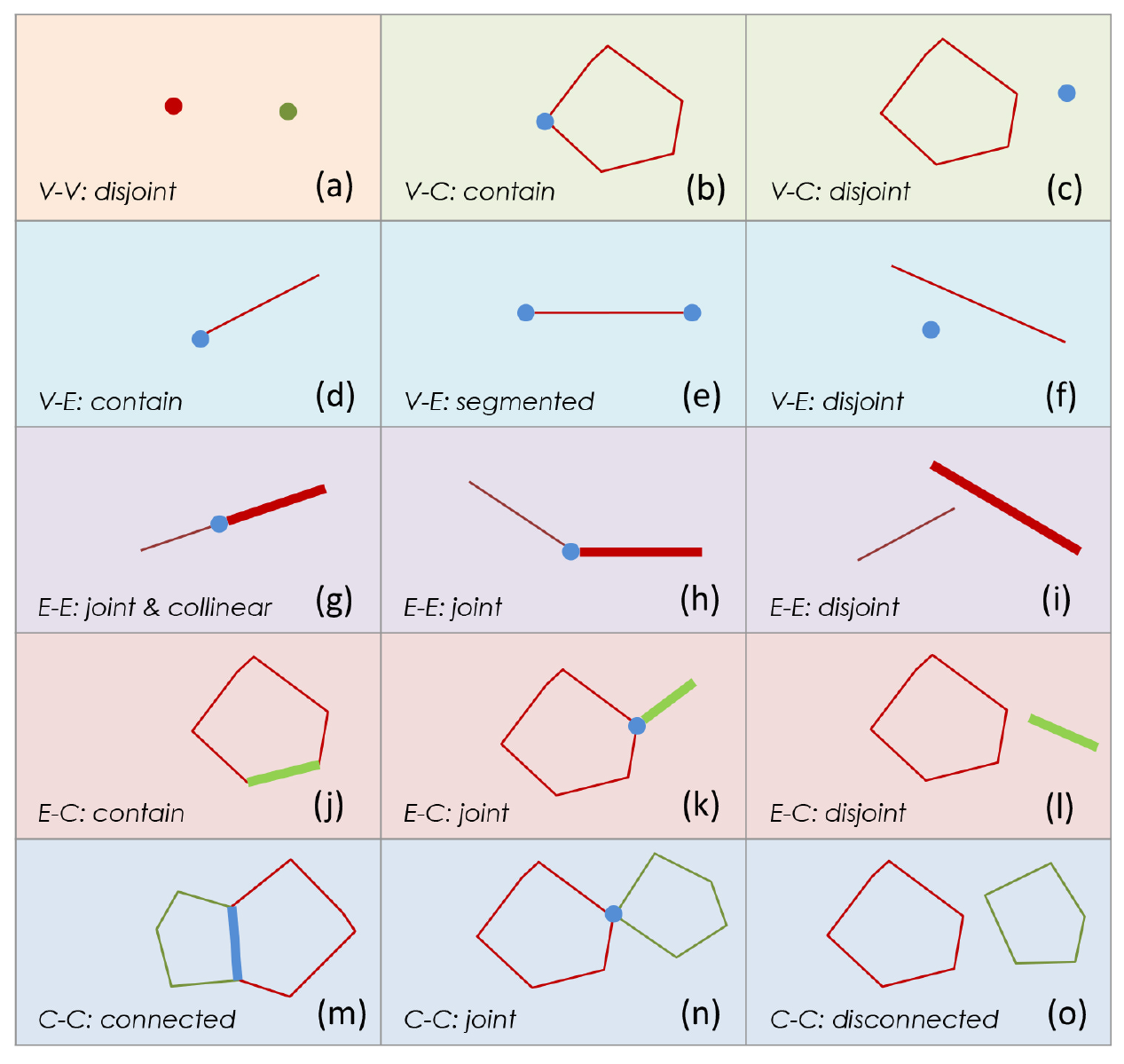}}
\caption{General spatial topological relations among particles in full-OACD.}
\label{Fig6}
\end{figure*}

\subsection{Spatial topology between particles}

There are six types of spatial combinations for particles: vertex-vertex
(V-V), vertex-edge (V-E), vertex-cell (V-C), edge-edge (E-E), edge-cell (E-C),
and cell-cell (C-C), and their relations are typically equal, joint, disjoint,
and others, see examples in Fig.6. These particle-particle relations also
underlie the further topological analysis of complexes.

\subsubsection{Vertex-Vertex (V-V) relations}

In terms of theorem \ref{Theorem_unique_vertex}, V-V relations between
$\varphi_{1}$ and $\varphi_{2}$ are quite simple -- either \emph{equal}, i.e.,
$\cap(\varphi_{1},\varphi_{2})=\varphi_{1}$, or \emph{disjoint}, i.e.,
$\cap(\varphi_{1},\varphi_{2})=\varnothing$, see Fig.6a.

\begin{proposition}
If $\varphi_{1}$ and $\varphi_{2}$ are two vertexes, then%
\begin{align}
\cap(\varphi_{1},\varphi_{2})  & =\varphi_{1}\Leftrightarrow\varphi
_{1}=\varphi_{2}\\
\cap(\varphi_{1},\varphi_{2})  & =\varnothing\Leftrightarrow\varphi_{1}%
\neq\varphi_{2}\nonumber
\end{align}

\end{proposition}

Because we have proved that chromatic codes are unique in OACD, the `equal'
relation, i.e., two particles are completely overlap, is easy to determine --
two particles are topologically equal, if and only if they are equal in codes,
that is, $\Omega_{1}$ $equal$ $\Omega_{2}\Leftrightarrow\Omega_{1}=\Omega_{2}%
$. Note, because a full-OACD is a type of spatial tessellation, meaning it
contains neither gaps nor overlaps, therefore topologically, it also does not
contain any two equal particles.

\subsubsection{Vertex-Edge (V-E) relations}

Typically an edge \emph{contains} two ends, and hence a V-E relation is that a
vertex is one of the ends of the edge (Fig.6d), that is, $\cap(\eta
,\varphi)=\varphi$; otherwise they are \emph{disjoint}, that is, $\cap
(\eta,\varphi)=\varnothing$, see Fig.6f.

\begin{proposition}
\label{Proposition_VE_joint}Given an edge $\eta$ and a vertex $\varphi$,%
\begin{align}
\cap(\eta,\varphi)  & =\varphi\Leftrightarrow\delta(\eta,\varphi)\leq2\\
\cap(\eta,\varphi)  & =\varnothing\Leftrightarrow\delta(\eta,\varphi
)>2\nonumber
\end{align}

\begin{proof}
Case(1) $\varphi$ is a 2-I vertex.

Suppose $\varphi$ is generated by $pb\left\langle i,j\right\rangle $ and
$pb\left\langle u,v\right\rangle $, and hence with a code
\begin{equation}
\varphi=(x_{i}^{a_{1}+\frac{1}{2}},x_{j}^{a_{1}+\frac{1}{2}},x_{u}%
^{a_{2}+\frac{1}{2}},x_{v}^{a_{2}+\frac{1}{2}})\cup(X_{others}^{A}%
)\label{Eq_VE1}%
\end{equation}

Case (1.1). $\eta$ is an edge generated from bisector $pb\left\langle
g,h\right\rangle $, and $g,h\notin\{i,j,u,v\}$. The codes of $\varphi$ and
$\eta$ then can be rewritten to $(x_{i}^{a_{1}+\frac{1}{2}}$, $x_{j}%
^{a_{1}+\frac{1}{2}}$, $x_{u}^{a_{2}+\frac{1}{2}}$, $x_{v}^{a_{2}+\frac{1}{2}%
}$, $x_{g}^{a_{3}}$, $x_{h}^{a_{4}})\cup(X_{others}^{A})$ and $(x_{i}^{b_{1}}%
$, $x_{j}^{b_{2}}$, $x_{u}^{b_{3}}$, $x_{v}^{b_{4}}$, $x_{g}^{b_{5}+\frac
{1}{2}}$, $x_{h}^{b_{5}+\frac{1}{2}})\cup(X_{others}^{B})$, respectively,
where $a_{i}$ and $b_{i}$ are both integers, and $A$ and $B$ are both integer
tuples. Because for chromatic distance at each component: $\delta
(x_{i})=|x_{i}^{a_{1}+\frac{1}{2}}-x_{i}^{b_{1}}|\geq\frac{1}{2}$,
$\delta(x_{j})=|x_{j}^{a_{1}+\frac{1}{2}}-x_{j}^{b_{2}}|\geq\frac{1}{2}$,
$\delta(x_{u})=|x_{u}^{a_{2}+\frac{1}{2}}-x_{u}^{b_{3}}|\geq\frac{1}{2}$,
$\delta(x_{v})=|x_{v}^{a_{2}+\frac{1}{2}}-x_{v}^{b_{4}}|\geq\frac{1}{2}$,
$\delta(x_{g})=|x_{g}^{a_{3}}-x_{g}^{b_{5}+\frac{1}{2}}|\geq\frac{1}{2}$,
$\delta(x_{h})=|x_{h}^{a_{4}}-x_{h}^{b_{5}+\frac{1}{2}}|\geq\frac{1}{2}$, and
$\delta(X_{other})=|X_{others}^{A}-X_{others}^{B}|\geq0$, then we know that
$\delta(\varphi,\eta)\geq3$.

Case (1.2). $\eta$ is an edge generated from the bisector involving one of
points $i$, $j$, $u$, and $v$, for example, $pb\left\langle i,k\right\rangle
$, then $\varphi$ and $\eta$ can be rewritten as $\varphi=(x_{i}^{a_{1}%
+\frac{1}{2}}$, $x_{j}^{a_{1}+\frac{1}{2}}$, $x_{u}^{a_{2}+\frac{1}{2}}$,
$x_{v}^{a_{2}+\frac{1}{2}}$, $x_{k}^{a_{3}})\cup(X_{others}^{A})$ and
$\eta=(x_{i}^{b_{1}+\frac{1}{2}}$, $x_{j}^{b_{2}}$, $x_{u}^{b_{3}}$,
$x_{v}^{b_{4}}$, $x_{k}^{b_{1}+\frac{1}{2}})\cup(X_{others}^{B})$. Since
distances at $x_{j}$, $x_{u}$, $x_{v}$, and $x_{k}$ are all greater than
$\frac{1}{2}$, we know that $\delta(\varphi,\eta)\geq2$.

Case (1.3). $\eta$ is an edge generated from the bisectors involving two
points in $i$, $j$, $u$, and $v$, but not the two who generate the vertex
$\varphi$, for example, $pb\left\langle i,u\right\rangle $, then $\eta$ can be
rewritten to $(x_{i}^{b_{1}+\frac{1}{2}}$, $x_{j}^{b_{2}}$, $x_{u}%
^{b_{1}+\frac{1}{2}}$, $x_{v}^{b_{3}})\cup(X_{others}^{B})$. Then similar to
the above cases (1) and (2), we know that $\delta(\varphi,\eta)\geq1$.
However, to reach the minimum $\delta=1$, it must be that $\delta(x_{i})=0$
and $\delta(x_{u})=0$, so that $a_{1}=b_{1}$ and $a_{2}=b_{1}$, then we have
$a_{1}=a_{2}$. But according to the bases of 2-I vertex (Property
\ref{P_2I_Base}), it is impossible that $a_{1}=a_{2}$, so we get $\delta$
cannot reach the minimum value 1, and hence $\delta\geq2$. To reach the new
minimum $\delta=2$, it must be $\delta(x_{i})=1$ and $\delta(x_{u})=0$, or
vice versa. This will lead to equations $a_{1}=a_{2}+1$ or $a_{2}=a_{1}+1$,
but they are impossible in terms of Property \ref{P_2I_Base}. Consequently,
$\delta\geq3$.

Case (1.4) $\eta$ is an edge generated from the bisectors, for example,
$pb\left\langle i,j\right\rangle $, which also make $\varphi$, then $\eta$'s
code should be%
\begin{equation}
\eta=(x_{i}^{b_{1}+\frac{1}{2}},x_{j}^{b_{1}+\frac{1}{2}},x_{u}^{b_{2}}%
,x_{v}^{b_{3}})\cup(X_{others}^{B})\label{Eq_VE2}%
\end{equation}
Then comparing Eq.(\ref{Eq_VE1}) and Eq.(\ref{Eq_VE2}), we know the
$\delta(\varphi,\eta)\geq1$. To reach the minimum value 1, it must be
$\delta(x_{u})=|a_{2}+\frac{1}{2}-b_{2}|=\frac{1}{2}$ and $\delta
(x_{v})=|a_{2}+\frac{1}{2}-b_{3}|=\frac{1}{2}$. Then we have two possible
solutions $%
\genfrac{\{}{.}{0pt}{}{a_{2}=b_{2}}{a_{2}=b_{2}-1}%
$ and $%
\genfrac{\{}{.}{0pt}{}{a_{2}=b_{3}}{a_{2}=b_{3}-1}%
$. The only allowed solution combinations are either $a_{2}=b_{2}$ and
$a_{2}=b_{3}-1$, or $a_{2}=b_{3}$ and $a_{2}=b_{2}-1 $. This indicates that
either the case $b_{2}=b_{3}-1$ or case $b_{3}=b_{2}-1 $. However, the two
cases are just the two edges which are on each side of $pb\left\langle
u,v\right\rangle $ and joint to $\varphi$. And in this case $\gamma
(\varphi,\eta)=2$.

Because $\delta(x_{i})=\delta(x_{j})\geq0$ and when they are both 0,
$\delta(\varphi,\eta)=1$. Therefore if they are both not 0, then $\delta
(x_{i},x_{j})\geq2\Rightarrow\delta(\varphi,\eta)\geq3$. So if in this case to
reach $\delta(\varphi,\eta)=2$, it must be $\delta(x_{i},x_{j})=0 $,
$\delta(x_{u},x_{v})=1$, and $\delta(X_{others})=1$. Because all components in
$X_{others}$ are integers, therefore $\delta(X_{others})$ is either 0 or
$\geq2$ but impossible to be 1. As a result, $\delta(\varphi,\eta)=2$ is
unable to reach, that is, $\delta(\varphi,\eta)>2$.

Case(2) $\varphi$ is a 3-I vertex.

Suppose $\varphi$ is generated by $pb\left\langle i,j\right\rangle $,
$pb\left\langle j,k\right\rangle $, $pb\left\langle k,i\right\rangle $, and
hence with a code $\varphi=(x_{i}^{a_{1}},x_{j}^{a_{1}},x_{k}^{a_{1}})$.

Case (2.1) $\eta$ is an edge generated from the bisector $pb\left\langle
u,v\right\rangle $, and $u,v\notin\{i,j,k\}$, then $\varphi$ and $\eta$ can be
rewritten as $\varphi=(x_{i}^{a_{1}}$, $x_{j}^{a_{1}}$, $x_{k}^{a_{1}}$,
$x_{u}^{a_{2}}$, $x_{v}^{a_{3}})\cup(X_{others}^{A})$ and $\eta=(x_{i}^{b_{1}%
}$, $x_{j}^{b_{2}}$, $x_{k}^{b_{3}}$, $x_{u}^{b_{4}+\frac{1}{2}}$,
$x_{v}^{b_{4}+\frac{1}{2}})\cup(X_{others}^{B})$. Because $b_{1}\neq b_{2}\neq
b_{3}$, then without loss of generality, assume $b_{2}=b_{1}+\Delta_{1}$,
$\Delta_{1}\geq1$, and $b_{3}=b_{1}+\Delta_{2}$, $\Delta_{2}\geq1$. If assume
$\delta(x_{i})=|a_{1}-b_{1}|=\Delta\geq0$, then $\delta(x_{j})=|a_{1}%
-b_{2}|=\Delta+\Delta_{1}\geq1$, $\delta(x_{k})=|a_{1}-b_{3}|=\Delta
+\Delta_{2}\geq1$, then we have $\delta(x_{i},x_{j},x_{k})\geq2$. In addition,
$\delta(x_{u})\geq\frac{1}{2}$ and $\delta(x_{v})\geq\frac{1}{2}$, therefore
$\delta(\varphi,\eta)\geq3$.

Case (2.2) $\eta$ is an edge generated from the bisector $pb\left\langle
i,u\right\rangle $, then $\varphi$ and $\eta$ can be rewritten as
$\varphi=(x_{i}^{a_{1}}$, $x_{j}^{a_{1}}$, $x_{k}^{a_{1}}$, $x_{u}^{a_{2}%
})\cup(X_{others}^{A}\backslash X_{removed}^{a_{1}-1,a_{1}+1})$ and
$\eta=(x_{i}^{b_{1}+\frac{1}{2}}$, $x_{j}^{b_{2}}$, $x_{k}^{b_{3}}$,
$x_{u}^{b_{1}+\frac{1}{2}})\cup(X_{others}^{B}\backslash X_{removed}%
^{b_{1},b_{1}+1})$, where $\backslash X_{removed}$ means some components, such
as $a_{1}-1$ and $a_{1}+1$, are excluded from $X_{others}$, in terms of the
bases of vertex and edge (Property \ref{P_3I_Base} and \ref{P_Edge_Base}).
Assume $b_{2}=b_{3}+\Delta$, then we have $\delta(x_{j},x_{k})\geq1$,
$\delta(x_{i})\geq\frac{1}{2}$, $\delta(x_{u})\geq\frac{1}{2}$. To reach the
minimum $\delta=2$, it must be $\delta(x_{i})=|a_{1}-b_{1}-\frac{1}{2}%
|=\frac{1}{2}$ and $\delta(x_{u})=|a_{2}-b_{1}-\frac{1}{2}|=\frac{1}{2}$. The
solutions are $%
\genfrac{\{}{.}{0pt}{}{a_{1}=b_{1}}{a_{1}=b_{1}+1}%
$ and $%
\genfrac{\{}{.}{0pt}{}{a_{2}=b_{1}}{a_{2}=b_{1}+1}
$. Then we have four combinations: (1) $a_{1}=b_{1}$, $a_{2}=b_{1}$ (2)
$a_{1}=b_{1}$, $a_{2}=b_{1}+1$; (3) $a_{1}=b_{1}+1$, $a_{2}=b_{1}$; (4)
$a_{1}=b_{1}+1$, $a_{2}=b_{1}+1$, but all these combinations are not allowed
because they will lead to impossible equations such that $a_{1}=a_{2}$,
$a_{2}=a_{1}-1$, or $a_{2}=a_{1}+1$, the components that have been removed. As
a result, $\delta$ is impossible to reach 2 and hence $\delta$ $\geq3$.

Case (2.3) $\eta$ is an edge generated from the bisector, say $pb\left\langle
i,j\right\rangle $, which is one of the three bisectors generating the
$\varphi$, then $\varphi$ and $\eta$ can be rewritten as $\varphi
=(x_{i}^{a_{1}}$, $x_{j}^{a_{1}}$, $x_{k}^{a_{1}})\cup(X_{others}^{A})$ and
$\eta=(x_{i}^{b_{1}+\frac{1}{2}}$, $x_{j}^{b_{1}+\frac{1}{2}}$, $x_{k}^{b_{2}%
})\cup(X_{others}^{B}\backslash X_{removed}^{b_{1},b_{1}+1})$. Using the
similar analysis in Case (2.2) , easy to know that $\delta\geq2$, and
$\delta=2$ only when case (2.3.1) $a_{1}=b_{1}$, $b_{2}=b_{1}-1$ or case
(2.3.2) $a_{1}=b_{1}+1$, $b_{2}=b_{1}+2 $. For the two cases, the two edges
bear codes $\eta_{1}=(x_{i}^{a_{1}+\frac{1}{2}}$, $x_{j}^{a_{1}+\frac{1}{2}}$,
$x_{k}^{a_{1}-1})$ and $\eta_{2}=(x_{i}^{a_{1}-\frac{1}{2}}$, $x_{j}%
^{a_{1}-\frac{1}{2}}$, $x_{k}^{a_{1}+1})$. In terms of
Eq.\ref{Eq_3ICodes_edges}, the two edges are just the ones that are in
bisector $pb\left\langle i,j\right\rangle $ and joint to a $\varphi^{3I}$. And
in these cases, $\gamma(\varphi,\eta)=3$.

Based on the above cases (1) and (2), we conclude that $\delta=1$ or
$\delta=2$ are the only cases that an edge contains an either 2-I or 3-I
vertex. For other cases, they must be disjoint.
\end{proof}
\end{proposition}

Given an edge, an useful function is to calculate all possible chromatic codes
of vertexes contained by the edge. Function $E2V(\eta)$ returns all contained
vertexes, and in particular, $E2V(\eta,2I)$ and $E2V(\eta,3I)$ return all 2-I
and 3-I vertexes, respectively.

\begin{notation}
The procedure of $E2V(\eta,2I)$:

Let $\eta$ is an edge with code $(x_{i}^{z+\frac{1}{2}},x_{j}^{z+\frac{1}{2}%
})\cup(X_{others})$, where $X_{others}=%
\mathbb{N}
\backslash\{z,z+1\}$. (1) Find the minimum component $w$ in $X_{others}$,
assume it is $x_{u}^{w}$; (2) Find $w+1$: if found, assume it is $x_{v}^{w+1}
$, then change $x_{u}^{w}$ and $x_{v}^{w+1}$ both to $w+\frac{1}{2}$ to form a
2-I vertex ($x_{i}^{z+\frac{1}{2}},x_{j}^{z+\frac{1}{2}},x_{u}^{w+\frac{1}{2}%
},x_{v}^{w+\frac{1}{2}}$); if not found, let $w=w+1$ and repeat (1) and (2)
until $w=n-2$.
\end{notation}

Because $(X_{others})$ can be partitioned into two parts $%
\mathbb{N}
_{1}=%
\mathbb{N}
\lbrack0,z-1]$ and $%
\mathbb{N}
_{2}=%
\mathbb{N}
\lbrack z+2,n-1]$, if $z\neq0$ and $z\neq n-2$, given any a component pair
such as $(x_{u}^{w},x_{v}^{w+1})$ in $%
\mathbb{N}
_{1}$ or $%
\mathbb{N}
_{2}$, it corresponds an edge with codes $(x_{i}^{z+\frac{1}{2}}%
,x_{j}^{z+\frac{1}{2}},x_{u}^{w+\frac{1}{2}},x_{v}^{w+\frac{1}{2}})$. Because
there are $z-1$ such component pairs in $%
\mathbb{N}
_{1}$, and $n-z-3$ pairs in $%
\mathbb{N}
_{2}$, therefore total $n-4$ available pairs. If $z\neq0$ or $z\neq n-2$, then
either $%
\mathbb{N}
_{1}$ or $%
\mathbb{N}
_{2}$ will be empty and the other will contain $n-3$ available pairs. For
example, for edge $(0,\frac{7}{2},5,1,2,\frac{7}{2})$ with $z=3$ (see edge
$07A247$ in Fig.3b), and hence it has $n-4=2$ available component pairs. We
first found $(0,1)$ and next $(1,2)$, and then they form two 2-I vertexes
$(\frac{1}{2},\frac{7}{2},5,\frac{1}{2},2,\frac{7}{2})$ (see vertex $17A147$
in Fig.3b) and $(0,\frac{7}{2},5,\frac{3}{2},\frac{3}{2},\frac{7}{2})$, respectively.

\begin{notation}
The procedure of $E2V(\eta,3I)$:

Let $\eta$ is an edge with code $(x_{i}^{z+\frac{1}{2}},x_{j}^{z+\frac{1}{2}%
})\cup(X_{others})$, where $X_{others}=%
\mathbb{N}
\backslash\{z,z+1\}$. (1) Find the minimum component $w$ in $(X_{others})$,
assume it is $x_{k}^{w}$; (2) If $e=(2z+1+w)\equiv0(\operatorname{mod}3)$ and
$e\notin(X_{others})$, then change $x_{i}^{z+\frac{1}{2}}$, $x_{j}^{z+\frac
{1}{2}}$, and $x_{k}^{w}$ to $\frac{e}{3}$ to form a 3-I vertex $(x_{i}%
^{\frac{e}{3}},x_{j}^{\frac{e}{3}},x_{k}^{\frac{e}{3}})$; (3) Let $w=w+1$ and
repeat (2) until $w=n-1$.
\end{notation}

Because $e\equiv0(\operatorname{mod}3)$ and $z=%
\mathbb{N}
\lbrack0,n-2]$, therefore $w=3m-2z-1$, with condition that $w=%
\mathbb{N}
\lbrack0,n-1]\backslash\{z,z+1\}$ and $m\in%
\mathbb{N}
\lbrack1,\frac{2}{3}(n+1)]$. For example, for edge $(2,3,\frac{9}{2}%
,0,1,\frac{9}{2})$ with $z=4$ (edge 469029 in Fig.3b), only $m=4$ and $w=3$
will form a 3-I vertex $(2,4,4,0,1,4)$ (see vertex $488028$ in Fig.3b).

Therefore the contain relation between a vertex $\varphi$ and an edge $\eta$
can be also determined by checking if $\varphi\in E2V(\eta)$.

Another V-E relation is that two vertexes are exactly the two ends of an edge,
called they are \emph{segmented} (Fig.6e), that is, $\cap(\eta,\varphi
_{1},\varphi_{2})=\{\varphi_{1},\varphi_{2}\}$. This relation is equivalent to
two vertexes share an edge.

\begin{proposition}
\label{Proposition_VE_segmented}Given an edge $\eta$ and two vertexes
$\varphi_{1}$ and $\varphi_{2}$, which could be both 2-I, or 3-I, or one is
2-I and the other is 3-I, then
\begin{align}
\cap(\eta,\varphi_{1}^{2I},\varphi_{2}^{2I})  & =\{\varphi_{1}^{2I}%
,\varphi_{2}^{2I}\}\Leftrightarrow\delta(\varphi_{1}^{2I},\varphi_{2}%
^{2I})=2\label{Proposition_VE_2I2I}\\
\cap(\eta,\varphi_{1}^{2I},\varphi_{2}^{3I})  & =\{\varphi_{1}^{2I}%
,\varphi_{2}^{3I}\}\Leftrightarrow\delta(\varphi_{1}^{2I},\varphi_{2}%
^{3I})=3\label{Proposition_VE_2I3I}\\
\cap(\eta,\varphi_{1}^{3I},\varphi_{2}^{3I})  & =\{\varphi_{1}^{3I}%
,\varphi_{2}^{3I}\}\Leftrightarrow\delta(\varphi_{1}^{3I},\varphi_{2}%
^{3I})=4\label{Proposition_VE_3I3I}%
\end{align}

\end{proposition}

\begin{proof}
Case (1) If $\varphi_{1}$ and $\varphi_{2}$ are both 2-I vertexes, then
according to Property \ref{P_2I_Codes} and \ref{P_2I_Unit}, it is easy to know
that $\delta(\varphi_{1},\varphi_{2})=2$. Below we prove that if
$\delta(\varphi_{1},\varphi_{2})=2$, then they are segmented.

Case (1.1) If $\varphi_{1}^{2I}$ is generated by $pb\left\langle
i,j\right\rangle $ and $pb\left\langle u,v\right\rangle $ and hence with a
code $(x_{i}^{a_{1}+\frac{1}{2}}$, $x_{j}^{a_{1}+\frac{1}{2}}$, $x_{u}%
^{a_{2}+\frac{1}{2}}$, $x_{v}^{a_{2}+\frac{1}{2}}$, $x_{e}^{a_{3}}$,
$x_{f}^{a_{4}}$, $x_{g}^{a_{5}}$, $x_{h}^{a_{6}})$ and $\varphi_{2}^{2I}$ is
generated by $pb\left\langle e,f\right\rangle $ and $pb\left\langle
g,h\right\rangle $ and hence with a code $(x_{i}^{b_{1}}$, $x_{j}^{b_{2}}$,
$x_{u}^{b_{3}}$, $x_{v}^{b_{4}}$, $x_{e}^{b_{5}+\frac{1}{2}}$, $x_{f}%
^{b_{5}+\frac{1}{2}}$, $x_{g}^{b_{6}+\frac{1}{2}}$, $x_{h}^{b_{6}+\frac{1}{2}%
})$. Apparently their $\delta\geq4$.

Case (1.2) If they share one point, say $e=i$, then their codes should be
$(x_{i}^{a_{1}+\frac{1}{2}}$, $x_{j}^{a_{1}+\frac{1}{2}}$, $x_{u}^{a_{2}%
+\frac{1}{2}}$, $x_{v}^{a_{2}+\frac{1}{2}}$, $x_{f}^{a_{3}}$, $x_{g}^{a_{4}}$,
$x_{h}^{a_{5}})$ and $(x_{i}^{b_{1}+\frac{1}{2}}$, $x_{j}^{b_{2}}$,
$x_{u}^{b_{3}}$, $x_{v}^{b_{4}}$, $x_{f}^{b_{1}+\frac{1}{2}}$, $x_{g}%
^{b_{5}+\frac{1}{2}}$, $x_{h}^{b_{5}+\frac{1}{2}})$, so $\delta\geq3$.

Case (1.3) If they share two points, say $e=i$ and $f=u$, then their codes
should be $(x_{i}^{a_{1}+\frac{1}{2}}$, $x_{j}^{a_{1}+\frac{1}{2}}$,
$x_{u}^{a_{2}+\frac{1}{2}}$, $x_{v}^{a_{2}+\frac{1}{2}}$, $x_{g}^{a_{3}}$,
$x_{h}^{a_{4}})$ and $(x_{i}^{b_{1}+\frac{1}{2}}$, $x_{j}^{b_{2}}$,
$x_{u}^{b_{1}+\frac{1}{2}}$, $x_{v}^{b_{3}}$, $x_{g}^{b_{4}+\frac{1}{2}}$,
$x_{h}^{b_{4}+\frac{1}{2}})$, so $\delta\geq2$. To make $\delta=2$, it must be
$\delta(x_{i})=\delta(x_{u})=0$, indicating that $a_{1}=a_{2}$, an impossible
case. Therefore $\delta\geq3$.

Case (1.4) If they share two points, say $e=i$ and $f=j$, then their codes
should be $(x_{i}^{a_{1}+\frac{1}{2}}$, $x_{j}^{a_{1}+\frac{1}{2}}$,
$x_{u}^{a_{2}+\frac{1}{2}}$, $x_{v}^{a_{2}+\frac{1}{2}}$, $x_{g}^{a_{3}}$,
$x_{h}^{a_{4}})\cup(X_{others}^{A}\backslash X_{removed}^{a_{1},a_{2}%
,a_{1}+1,a_{2}+1})$ and $(x_{i}^{b_{1}+\frac{1}{2}}$, $x_{j}^{b_{1}+\frac
{1}{2}}$, $x_{u}^{b_{2}}$, $x_{v}^{b_{3}}$, $x_{g}^{b_{4}+\frac{1}{2}}$,
$x_{h}^{b_{4}+\frac{1}{2}})\cup(X_{others}^{B}\backslash X_{removed}%
^{b_{1},b_{4},b_{1}+1,b_{4}+1})$. Theoretically $\delta\geq2$. To reach 2, it
must be $\delta(x_{i})=\delta(x_{j})=0\Rightarrow a_{1}=b_{1}$, $\delta
(x_{u})=\frac{1}{2}\Rightarrow%
\genfrac{\{}{.}{0pt}{}{a_{2}=b_{2}}{a_{2}=b_{2}-1}%
$, $\delta(x_{v})=\frac{1}{2}\Rightarrow%
\genfrac{\{}{.}{0pt}{}{a_{2}=b_{3}}{a_{2}=b_{3}-1}%
$, $\delta(x_{g})=\frac{1}{2}\Rightarrow%
\genfrac{\{}{.}{0pt}{}{a_{3}=b_{4}}{a_{3}=b_{4}+1}%
$, $\delta(x_{h})=\frac{1}{2}\Rightarrow%
\genfrac{\{}{.}{0pt}{}{a_{4}=b_{4}}{a_{4}=b_{4}+1}%
$. Then we have 16 combinations for solutions, but because $a_{3}\neq a_{4}$,
$b_{2}\neq b_{3}$, as well as some components have been removed, therefore
only the below four solutions are allowed.
\begin{equation}
\left\{
\begin{array}
[c]{c}%
a_{2}=b_{2},a_{2}=b_{3}-1,a_{3}=b_{4},a_{4}=b_{4}+1\\
a_{2}=b_{2},a_{2}=b_{3}-1,a_{3}=b_{4}+1,a_{4}=b_{4}\\
a_{2}=b_{2}-1,a_{2}=b_{3},a_{3}=b_{4},a_{4}=b_{4}+1\\
a_{2}=b_{2}-1,a_{2}=b_{3},a_{3}=b_{4}+1,a_{4}=b_{4}%
\end{array}
\right. \label{Eq_VE_segmented_4solutions}%
\end{equation}
Also in a 2-I unit, $\varphi_{1}^{2I}$ should link to two edges with codes%
\begin{align}
& (x_{i}^{a_{1}+\frac{1}{2}},x_{j}^{a_{1}+\frac{1}{2}},x_{u}^{a_{2}}%
,x_{v}^{a_{2}+1},x_{g}^{a_{3}},x_{h}^{a_{4}})\label{Eq_VE_segmented_2edges1}\\
& (x_{i}^{a_{1}+\frac{1}{2}},x_{j}^{a_{1}+\frac{1}{2}},x_{u}^{a_{2}+1}%
,x_{v}^{a_{2}},x_{g}^{a_{3}},x_{h}^{a_{4}})\nonumber
\end{align}
and $\varphi_{2}^{2I}$ should link to two edges with codes%
\begin{align}
& (x_{i}^{b_{1}+\frac{1}{2}},x_{j}^{b_{1}+\frac{1}{2}},x_{u}^{b_{2}}%
,x_{v}^{b_{3}},x_{g}^{b_{4}},x_{h}^{b_{4}+1})\label{Eq_VE_segmented_2edges2}\\
& (x_{i}^{b_{1}+\frac{1}{2}},x_{j}^{b_{1}+\frac{1}{2}},x_{u}^{b_{2}}%
,x_{v}^{b_{3}},x_{g}^{b_{4}+1},x_{h}^{b_{4}})\nonumber
\end{align}

Given a solution in Eq.\ref{Eq_VE_segmented_4solutions}, we can always find
two edges, in which one is from $\varphi_{1}^{2I}$
(Eq.\ref{Eq_VE_segmented_2edges1}) and the other is from $\varphi_{2}^{2I}$
(Eq.\ref{Eq_VE_segmented_2edges2}) are same, and therefore we get $\varphi
_{1}^{2I}$ and $\varphi_{2}^{2I}$ are the two ends of an edge.

Case (1.5) If they share three points, say $e=i$, $f=j$, and $g=u$, then their
codes should be $(x_{i}^{a_{1}+\frac{1}{2}}$, $x_{j}^{a_{1}+\frac{1}{2}}$,
$x_{u}^{a_{2}+\frac{1}{2}}$, $x_{v}^{a_{2}+\frac{1}{2}}$, $x_{h}^{a_{3}}%
)\cup(X_{others}^{A}\backslash X_{removed}^{a_{1},a_{2},a_{1}+1,a_{2}+1})$ and
$(x_{i}^{b_{1}+\frac{1}{2}}$, $x_{j}^{b_{1}+\frac{1}{2}}$, $x_{u}^{b_{2}%
+\frac{1}{2}}$, $x_{v}^{b_{3}}$, $x_{h}^{b_{2}+\frac{1}{2}})\cup
(X_{others}^{B}\backslash X_{removed}^{b_{1},b_{2},b_{1}+1,b_{2}+1})$.
Theoretically $\delta\geq1$. To reach 1, it must be $a_{1}=b_{1}$,
$a_{2}=b_{2}$, $%
\genfrac{\{}{.}{0pt}{}{a_{2}=b_{3}}{a_{2}=b_{3}-1}%
$, $%
\genfrac{\{}{.}{0pt}{}{a_{3}=b_{2}}{a_{3}=b_{2}+1}%
$. Because $a_{2}=b_{3}\Rightarrow b_{2}=b_{3}$, $a_{3}=b_{2}\Rightarrow
a_{2}=a_{3}$, $a_{2}=b_{3}-1\Rightarrow b_{3}=b_{2}+1$, $a_{3}=b_{2}%
+1\Rightarrow a_{3}=a_{2}+1$, but all these equations are impossible because
they have been already removed from their codes.

Then the next minimum $\delta$ should be 2. Because $\delta(x_{i}%
)=\delta(x_{j})\geq0$, they must be = 0 in this case, since if they = 1, then
$\delta\geq3$. Because $\delta(x_{v})$ and $\delta(x_{h})$ at least contribute
1 to $\delta$, then if $\delta=2$, then anther $\delta(x)$ should be from
$\delta(x_{u})$ or $\delta(X_{others})$. But if $\delta(X_{others})=1$, then
$\delta(x_{u})$ must be 0, and if $\delta(x_{u})=0$, then we have that
$\delta(x_{i})=\delta(x_{j})=$ $\delta(x_{u})=0$ and $\delta(x_{v})=$
$\delta(x_{h})=\frac{1}{2}$, an impossible case we just proved above.
Therefore we have $\delta(x_{u})=1\Rightarrow%
\genfrac{\{}{.}{0pt}{}{a_{2}=b_{2}+1}{a_{2}=b_{2}-1}%
$ and also $%
\genfrac{\{}{.}{0pt}{}{a_{2}=b_{3}}{a_{2}=b_{3}-1}%
$, $%
\genfrac{\{}{.}{0pt}{}{a_{3}=b_{2}}{a_{3}=b_{2}+1}%
$. We then have 8 combinations for these solutions such as%
\begin{equation}
\left\{
\begin{array}
[c]{c}%
a_{2}=b_{2}+1,a_{2}=b_{3},a_{3}=b_{2}\\
a_{2}=b_{2}+1,a_{2}=b_{3},a_{3}=b_{2}+1\\
a_{2}=b_{2}+1,a_{2}=b_{3}-1,a_{3}=b_{2}\\
a_{2}=b_{2}+1,a_{2}=b_{3}-1,a_{3}=b_{2}+1\\
a_{2}=b_{2}-1,a_{2}=b_{3},a_{3}=b_{2}\\
a_{2}=b_{2}-1,a_{2}=b_{3},a_{3}=b_{2}+1\\
a_{2}=b_{2}-1,a_{2}=b_{3}-1,a_{3}=b_{2}\\
a_{2}=b_{2}-1,a_{2}=b_{3}-1,a_{3}=b_{2}+1
\end{array}
\right.
\end{equation}
By checking these solutions we can always find some removed components at $i$,
$j$, $u$, $v$, and $h$, except for the solution that $a_{2}=b_{2}+1$,
$a_{2}=b_{3}-1$, $a_{3}=b_{2}$. Also $\varphi_{1}^{2I}$ should link to an edge
$\eta_{1}=(x_{i}^{a_{1}+\frac{1}{2}}$, $x_{j}^{a_{1}+\frac{1}{2}}$,
$x_{u}^{a_{2}}$, $x_{v}^{a_{2}+1}$, $x_{h}^{a_{3}})$ and $\varphi_{2}^{2I}$
should link to an edge $\eta_{2}=(x_{i}^{b_{1}+\frac{1}{2}}$, $x_{j}%
^{b_{1}+\frac{1}{2}}$, $x_{u}^{b_{2}+1}$, $x_{v}^{b_{3}}$, $x_{h}^{b_{2}})$.
Since $(a_{3}=b_{2}$, $a_{2}=b_{2}+1)\Rightarrow a_{3}=a_{2}-1$, $a_{2}%
=b_{3}-1\Rightarrow b_{3}=a_{2}+1$, $a_{2}=b_{2}+1\Rightarrow b_{2}=a_{2}-1$,
therefore $\eta_{1}=\eta_{2}=(x_{i}^{a_{1}+\frac{1}{2}}$, $x_{j}^{a_{1}%
+\frac{1}{2}}$, $x_{u}^{a_{2}}$, $x_{v}^{a_{2}+1}$, $x_{h}^{a_{2}-1})$,
indicating that the two vertexes are both the ends of the same edge.
\end{proof}

Above we proved that if chromatic distance between two 2-I vertexes is 2, then
they must be segmented with an edge. The other two cases
Eq.\ref{Proposition_VE_2I3I} and \ref{Proposition_VE_3I3I} can be proved in
the same way, but they are too long and hence not presented here.

Similarly, we can also use $E2V$ function to determine if two vertexes and one
edge are segmented, that is,%
\begin{equation}
\cap(\eta,\varphi_{1},\varphi_{2})=\{\varphi_{1},\varphi_{2}\}\Leftrightarrow
\varphi_{1},\varphi_{2}\in E2V(\eta)
\end{equation}

\subsubsection{Vertex-Cell (V-C) relations}

The relation between a vertex $\varphi$ and a cell $\zeta$ is that $\varphi$
is one of edge ends which are the boundaries of $\zeta$, called the cell
contains $\varphi^{3I}$, i.e., $\cap(\zeta,\varphi)=\varphi$, see Fig.6b;
otherwise, they are \emph{disjoint}, i.e., $\cap(\zeta,\varphi)=\varnothing$,
see Fig.6c.

\begin{proposition}
\label{Proposition_VC_joint}Given a vertex $\varphi$ and a cell $\zeta$, then%
\begin{align}
\cap(\zeta,\varphi)  & =\varphi\Leftrightarrow\delta(\zeta,\varphi)=2\\
\cap(\zeta,\varphi)  & =\varnothing\Leftrightarrow\delta(\zeta,\varphi
)>2\nonumber
\end{align}

\end{proposition}

\begin{proof}
From table \ref{Table1} we know that either for $\varphi^{2I}$ or
$\varphi^{3I}$, $\delta(\varphi,\zeta)=2$.

Case (1): Suppose a vertex $\varphi^{2I}$ bears a code $(x_{i}^{a_{1}+\frac
{1}{2}}$, $x_{j}^{a_{1}+\frac{1}{2}}$, $x_{u}^{a_{2}+\frac{1}{2}}$,
$x_{v}^{a_{2}+\frac{1}{2}})\cup(X_{others}^{A})$ and a cell $\zeta$ bears a
code $(x_{i}^{b_{1}}$, $x_{j}^{b_{2}}$, $x_{u}^{b_{3}}$, $x_{v}^{b_{4}}%
)\cup(X_{others}^{B})$, then $\delta(\varphi^{2I},\zeta)\geq2$. To reach 2, it
must be $|a_{1}+\frac{1}{2}-b_{1}|=\frac{1}{2}$, $|a_{1}+\frac{1}{2}%
-b_{2}|=\frac{1}{2}$, $|a_{2}+\frac{1}{2}-b_{3}|=\frac{1}{2}$, $|a_{2}%
+\frac{1}{2}-b_{4}|=\frac{1}{2}$, and $|X_{others}^{A}-X_{others}^{B}|=0$.
Therefore for each equation, the solutions should be the cases that $%
\genfrac{\{}{.}{0pt}{}{a_{1}=b_{1}}{a_{1}=b_{1}-1}%
$, $%
\genfrac{\{}{.}{0pt}{}{a_{1}=b_{2}}{a_{1}=b_{2}-1}%
$, $%
\genfrac{\{}{.}{0pt}{}{a_{2}=b_{3}}{a_{2}=b_{3}-1}%
$, and $%
\genfrac{\{}{.}{0pt}{}{a_{2}=b_{4}}{a_{2}=b_{4}-1}%
$, respectively. It is also known that $b_{1}\neq b_{2}\neq b_{3}\neq b_{4}$,
thus only the below four solutions are allowed to give $\delta(\varphi
^{2I},\zeta)=2$: (1) $a_{1}=b_{1}$, $b_{1}=b_{2}-1$, (2) $a_{1}=b_{2}$,
$b_{2}=b_{1}-1$, (3) $a_{2}=b_{3}$, $b_{3}=b_{4}-1$, and (4) $a_{2}=b_{4}$,
$b_{4}=b_{3}-1$. It is easy to know that the four solutions are just the four
cells around the vertex $\varphi^{2I}$ in a 2-I unit.

Case (2): suppose a vertex $\varphi^{3I}$ bears a code $(x_{i}^{a_{1}}%
,x_{j}^{a_{1}},x_{k}^{a_{1}})$ and a cell $\zeta$ bears a code $(x_{i}^{b_{1}%
},x_{j}^{b_{2}},x_{k}^{b_{3}})$, and let $|a_{1}-b_{1}|=\Delta_{1}$,
$|a_{1}-b_{2}|=\Delta_{2}$, $|a_{1}-b_{3}|=\Delta_{3}$. Although theoretically
$\delta(\varphi^{3I},\zeta)=\Delta_{1}+\Delta_{2}+\Delta_{3}$ $\geq0$, it is
impossible to reach 0 or even 1 because $b_{1}\neq b_{2}\neq b_{3}$. Therefore
the next minimum $\delta(\varphi^{3I},\zeta)$ is 2 and it should be given by
two of $\Delta_{1},\Delta_{2},$and $\Delta_{3}$ both = 1 and one of them = 0.
Let's assume $\Delta_{1}=0$ and hence $a_{1}=b_{1}$, then we have
$|b_{1}-b_{2}|=1$ and $|b_{1}-b_{3}|=1$. The two equations have solutions $%
\genfrac{\{}{.}{0pt}{}{b_{1}=b_{2}+1}{b_{1}=b_{2}-1}%
$and $%
\genfrac{\{}{.}{0pt}{}{b_{1}=b_{3}+1}{b_{1}=b_{3}-1}%
$, and hence only two solutions are allowed: $%
\genfrac{\{}{.}{0pt}{}{b_{1}=b_{2}+1,b_{1}=b_{3}-1}{b_{1}=b_{2}-1,b_{1}%
=b_{3}+1}%
$. If $\Delta_{2}=0$ or $\Delta_{3}=0$, we can get another four allowed
solutions which give $\delta(\varphi^{3I},\zeta)=2$, and in total we get six
solutions. Comparing these solutions to those cells (Eq.\ref{Eq_3ICodes_cells}%
) around the $\varphi^{3I}$ in a 3-I unit, we thus to know that if
$\delta(\varphi^{3I},\zeta)=2$, the cell must contain the $\varphi^{3I}$.
\end{proof}

\subsubsection{Edge-Edge (E-E) relations}

In 2-I and 3-I units there three types of relations between two edges, i.e.,
adjacent, opposite, and interval, and their chromatic distances could be 2, 3,
or 4, respectively, but actually all the three relations are topologically
same as two edges share an either 2-I or 3-I vertex as one of their ends. This
E-E\ relation is called the two edges are \emph{joint }with a vertex, i.e.,
$\cap(\eta_{1},\eta_{2})=\varphi$, which further has two types: (1) collinear,
denoted by $\overline{\eta_{1}\eta_{2}}$ (Fig.6g), and (2) not collinear
(Fig.6h). If two edges do not share any particles, they are \emph{disjoint}
(Fig.6i). The E-E collinear relation is easy to determine by using the below proposition.

\begin{proposition}
\label{Proposition_EE_collinear}Given two edges $\eta_{1}(x_{i}^{z_{1}%
+\frac{1}{2}},x_{j}^{z_{1}+\frac{1}{2}})$ and $\eta_{2}(x_{u}^{z_{2}+\frac
{1}{2}},x_{v}^{z_{2}+\frac{1}{2}})$,%
\begin{equation}
\overline{\eta_{1}\eta_{2}}\Leftrightarrow i=u,j=v
\end{equation}

\end{proposition}

Note that Proposition \ref{Proposition_EE_collinear} can only tell if two
edges are collinear, but two collinear edges may not be joint. A feasible
method to reason topological joint between two edges is using $E2V(\eta)$
function to calculate all possible vertexes contained by the two edges, and if
among them two vertexes are equal, then they are joint with this vertex. This
method, however, can only tell if two edges are possible to be joint, but in a
real full-OACD, they may not be joint because the joint vertex is hidden in
high-dimensional spaces. For example, the edge (36A038) and (25A058) in Fig.
4b both have an end vertex (44A048) and hence joint, but the vertex does not
emerge in the $%
\mathbb{R}
^{2}$ plane, so that the two edges appear disjoint. Similarly, if we only use
the chromatic distances listed in Table 1 to determine E-E relations, then
they may also lead to mistakes in $%
\mathbb{R}
^{2}$ plane due to the same reason. For the same example, in Fig.3b,
$\delta((36A038),(25A058))=2$, $\gamma((36A038),(25A058))=3$, indicating they
are joint with a 3-I vertex, i.e., the result (44A048) of $E2V$, but in the
given plane they are not joint. Therefore, the below proposition is true for
$OACD$ at $%
\mathbb{R}
^{n-1}$ space, which contains all possible edges, cells, and vertexes.

\begin{proposition}
\label{Proposition_EE_joint}Given two edges $\eta_{1}$ and $\eta_{2}$ in
$OACD(n,%
\mathbb{R}
^{n-1})$, let $(\delta,\gamma)=\\(\delta(\eta_{1},\eta_{2}),\gamma(\eta_{1}%
,\eta_{2}))$, then%
\begin{equation}
\cap(\eta_{1},\eta_{2})=\varphi\Leftrightarrow\left\{
\begin{array}
[c]{l}%
\varphi^{2I}\Leftrightarrow(\delta,\gamma)=(2,2)\vee(2,4)\\
\varphi^{3I}\Leftrightarrow(\delta,\gamma)=(2,3)\vee(3,2)\vee(4,3)
\end{array}
\right.
\end{equation}

\end{proposition}

\begin{proof}
Suppose $\eta_{1}$ is an edge generated from $pb\left\langle i,j\right\rangle
$ and hence with a code $(x_{i}^{a_{ij}+\frac{1}{2}}$, $x_{j}^{a_{ij}+\frac
{1}{2}})$, and $\eta_{2}$ is an edge generated from $pb\left\langle
u,v\right\rangle $ and hence with a code $(x_{u}^{b_{uv}+\frac{1}{2}}$,
$x_{v}^{b_{uv}+\frac{1}{2}})$.

Case (1) If $i\neq j\neq u\neq v$, then $\eta_{1}$ and $\eta_{2}$ can be
rewritten to $(x_{i}^{a_{ij}+\frac{1}{2}}$, $x_{j}^{a_{ij}+\frac{1}{2}}$,
$x_{u}^{a_{u}}$, $x_{v}^{a_{v}})\cup(X_{others}^{A})$ and $(x_{i}^{b_{i}}$,
$x_{j}^{b_{j}}$, $x_{u}^{b_{uv}+\frac{1}{2}}$, $x_{v}^{b_{uv}+\frac{1}{2}%
})\cup(X_{others}^{B})$, and if they are joint, then the joint vertex should
be 2-I. The minimum $\delta(\eta_{1},\eta_{2})=2$ if $\delta(x_{i})=$
$\delta(x_{j})=$ $\delta(x_{u})=\delta(x_{v})=\frac{1}{2}$, and $\delta
(X_{others})=0$. Therefore the solutions are $%
\genfrac{\{}{.}{0pt}{}{a_{ij}=b_{i}}{a_{ij}=b_{i}-1}%
$, $%
\genfrac{\{}{.}{0pt}{}{a_{ij}=b_{j}}{a_{ij}=b_{j}-1}%
$, $%
\genfrac{\{}{.}{0pt}{}{a_{u}=b_{uv}}{a_{u}=b_{uv}+1}%
$, and $%
\genfrac{\{}{.}{0pt}{}{a_{v}=b_{uv}}{a_{v}=b_{uv}+1}%
$, where only the below four solutions are allowed.%
\begin{equation}
\left\{
\begin{array}
[c]{c}%
a_{ij}=b_{i},a_{ij}=b_{j}-1,a_{u}=b_{uv},a_{v}=b_{uv}+1\\
a_{ij}=b_{i},a_{ij}=b_{j}-1,a_{u}=b_{uv}+1,a_{v}=b_{uv}\\
a_{ij}=b_{j},a_{ij}=b_{i}-1,a_{u}=b_{uv},a_{v}=b_{uv}+1\\
a_{ij}=b_{j},a_{ij}=b_{i}-1,a_{u}=b_{uv}+1,a_{v}=b_{uv}%
\end{array}
\right.
\end{equation}
For each solution, we substitute them into codes of $\eta_{1}$ and $\eta_{2}$,
and these substitutions, for example of the first solution, will make their
codes being $(x_{i}^{a_{ij}+\frac{1}{2}}$, $x_{j}^{a_{ij}+\frac{1}{2}}$,
$x_{u}^{a_{u}}$, $x_{v}^{a_{u}+1})$ and $(x_{i}^{a_{ij}}$, $x_{j}^{a_{ij}+1}$,
$x_{u}^{a_{u}+\frac{1}{2}}$, $x_{v}^{a_{u}+\frac{1}{2}})$. Easy to see that
the two edges are joint with a 2-I vertex $(x_{i}^{a_{ij}+\frac{1}{2}}$,
$x_{j}^{a_{ij}+\frac{1}{2}}$, $x_{u}^{a_{u}+\frac{1}{2}}$, $x_{v}^{a_{u}%
+\frac{1}{2}})$ and $\gamma(\eta_{1},\eta_{2})=4$, $\eta_{1}\ncong\eta_{2}$.

Since components in $X_{others}$ are always integers, $X_{others}\neq1$, then
to make $\delta=3$, it must be $\delta(x_{i},x_{j},x_{u},x_{v})=3$. Suppose
$\delta(x_{i})=\frac{i}{2}$, $\delta(x_{j})=\frac{j}{2}$, $\delta(x_{u}%
)=\frac{u}{2}$, $\delta(x_{v})=\frac{v}{2}$, where $i$, $j$, $u$, and $v$ are
odd numbers. Then $\frac{i}{2}+\frac{j}{2}+\frac{u}{2}+\frac{v}{2}%
=3\Rightarrow i+j+u+v=6$, then the only solution is $i=j=u=1$, $v=3$. But this
solution is impossible because it leads to $X_{others}\neq0\Rightarrow
\delta>3$. To make $\delta=4$, it must be two cases (1.1) $\delta(x_{i}%
,x_{j},x_{u},x_{v})=4$, $\delta(X_{others})=0$ or (1.2) $\delta(x_{i}%
,x_{j},x_{u},x_{v})=2$, $\delta(X_{others})=2$. Similarly, only the case (1.2)
is possible, and $\gamma=6$. And easy to check that edges in case (1.2) are
not joint to a vertex.

Case (2) If $i=u$ and $j=v$, then according to Proposition
\ref{Proposition_EE_collinear}, the two edges should be generated from the
same bisector $pb\left\langle i,j\right\rangle =pb\left\langle
u,v\right\rangle $. Therefore their codes can be rewritten to $(x_{i}%
^{a_{1}+\frac{1}{2}},x_{j}^{a_{1}+\frac{1}{2}})\cup(X_{others}^{A})$ and
$(x_{i}^{b_{1}+\frac{1}{2}},x_{j}^{b_{1}+\frac{1}{2}})\cup(X_{others}^{B})$.

Let us discuss all possible cases of $\delta(\eta_{1},\eta_{2})$ for
$D=a_{1}-b_{1}.$

Case (2.1) If $D=0$, then $\delta(\eta_{1},\eta_{2})=|X_{others}%
^{A}-X_{others}^{B}|$, and also $X_{others}^{B}$ is a permutation of
$X_{others}^{A}=%
\mathbb{N}
\backslash\{a_{1},a_{1}+1\}$. Therefore the possible values of $\delta$ will
be given by some components in $X_{others}$ changing their locations. Suppose
the number of such components is $m$.

Case (2.1.1) If $m=0$, then $\delta=0$, so this case can be excluded because
it makes the two edges are equal.

Case (2.1.2) If $m=2$, assuming the two components are $x_{g}^{a_{g}}$ and
$x_{h}^{a_{h}}$, then $\delta=2|a_{g}-a_{h}|$. Because $a_{g}\neq a_{h}$,
therefore $\delta$'s possible values will be even numbers.

If $\delta=2$, then it must be $|a_{g}-a_{h}|=1$. The two edges can be
rewritten to $(x_{i}^{a_{1}+\frac{1}{2}}$, $x_{j}^{a_{1}+\frac{1}{2}}$,
$x_{g}^{a_{g}}$, $x_{h}^{a_{g}+1})$ and $(x_{i}^{a_{1}+\frac{1}{2}}$,
$x_{j}^{a_{1}+\frac{1}{2}}$, $x_{g}^{a_{g}+1}$, $x_{h}^{a_{g}})$, indicating
they are joint to a 2-I vertex generated by $pb\left\langle i,j\right\rangle $
and $pb\left\langle g,h\right\rangle $. And in this case, $\gamma=2$,
$\eta_{1}\cong\eta_{2}$.

If $\delta=4$, then it must be $|a_{g}-a_{h}|=2$, and also $\gamma=2$, but in
this case, the two edges are not joint.

Case (2.1.3) If $m=3$, indicating that three components involve changing their
locations, we can express the three components as $X_{three}=(a_{2}%
,a_{2}+\Delta_{1},a_{2}+\Delta_{1}+\Delta_{2})$ with $\Delta_{1}\geq1$ and
$\Delta_{2}\geq1$. Easy to know that $X_{three}$ has 6 permutations, but in
which only two permutations involve changing all three components, that is,
$X_{three}^{1}=(a_{2}+\Delta_{1},a_{2}+\Delta_{1}+\Delta_{2},a_{2})$ and
$X_{three}^{2}=(a_{2}+\Delta_{1}+\Delta_{2},a_{2},a_{2}+\Delta_{1})$. Then
$\delta(X_{three},X_{three}^{1})=\delta(X_{three},X_{three}^{2})=2(\Delta
_{1}+\Delta_{2})$ $\geq4$, and only when $\Delta_{1}=\Delta_{2}=1$, $\delta
=4$, $\gamma=3$, $\eta_{1}\cong\eta_{2}$, but in this case, the two edges are
not joint.

Case (2.1.4) If $m=4$, $\delta=4$, then there are four components in
$X_{others}$ contributing 1 respectively to $\delta$, but in this case, the
two edges are not joint.

Case (2.1.5) If $m>4$, then $\delta\geq4$.

Case (2.2) If $D=1$, then $\delta(x_{i})=\delta(x_{j})=1$ $\Rightarrow
\delta\geq2$. Also $X_{others}^{A}=%
\mathbb{N}
\backslash\{a_{1},a_{1}+1\}$ and $X_{others}^{B}=%
\mathbb{N}
\backslash\{a_{1}-1,a_{1}\}$. For similar cases such as (2.1.1)-(2.1.4), we
know that $X_{others}^{A}$ includes $a_{1}-1$ but excludes $a_{1}+1$, and
$X_{others}^{B}$ includes $a_{1}+1$ but excludes $a_{1}-1$. Therefore we have
$|X_{others}^{A}-X_{others}^{B}|\geq2$. The $X_{others}^{A}$ and
$X_{others}^{B}$ can be further rewritten to $X_{others}^{A}=(x_{g}^{a_{1}%
-1},x_{h}^{a_{A}})\cup(X_{others^{\prime}}^{A^{\prime}})$ and $X_{others}%
^{B}=(x_{g}^{a_{B}},x_{h}^{a_{1}+1})\cup(X_{others^{\prime}}^{B^{\prime}})$.

Case (2.2.1) $g\neq h$. To reach $|X_{others}^{A}-X_{others}^{B}|=2$, it must
be $\delta(x_{g})=1$, $\delta(x_{h})=1$, and $\delta(x_{others^{\prime}})=0$.
Thus the solutions are $%
\genfrac{\{}{.}{0pt}{}{a_{B}=a_{1}}{a_{B}=a_{1}-2}%
$ and $%
\genfrac{\{}{.}{0pt}{}{a_{A}=a_{1}}{a_{A}=a_{1}+2}%
$. Because $a_{1}$ has been already excluded from both $X_{others}^{A}$ and
$X_{others}^{B}$, the only solution is $a_{B}=a_{1}-2$ and $a_{A}=a_{1}+2$.
However, this solution will make $X_{others^{\prime}}^{A^{\prime}}$ include
$a_{B}$ but exclude $a_{A}$, while $X_{others^{\prime}}^{B^{\prime}}$ include
$a_{A}$ but exclude $a_{B}$, implying $\delta(x_{others^{\prime}})>0$.
Therefore, $|X_{others}^{A}-X_{others}^{B}|$ is impossible to be 2, and hence
$\delta>4$.

Case (2.2.2) $g=h$. Then $|X_{others}^{A}-X_{others}^{B}|=2$ if
$|X_{others^{\prime}}^{A^{\prime}}-X_{others^{\prime}}^{B^{\prime}}|=0$. In
this case, $\delta=4$ and $\gamma=3$. The two edges turn to $(x_{i}%
^{a_{1}+\frac{1}{2}}$, $x_{j}^{a_{1}+\frac{1}{2}}$, $x_{g}^{a_{1}-1})$ and
$(x_{i}^{a_{1}-\frac{1}{2}}$, $x_{j}^{a_{1}-\frac{1}{2}}$, $x_{g}^{a_{1}+1})$.
We can find that $x_{i}+x_{j}+x_{g}=3a_{1}$, then using $E2V(\eta,3I)$ we know
that they are both joint to a 3-I vertex $(x_{i}^{\frac{a_{1}}{3}}%
,x_{j}^{\frac{a_{1}}{3}},x_{g}^{\frac{a_{1}}{3}})$. Also $D=1 $ always implies
$\eta_{1}\ncong\eta_{2}$.

Case (2.3) If $D\geq2$, then $\delta(x_{i})\geq2$, $\delta(x_{j})\geq2$ and
hence $\delta\geq4$. Using the similar analysis in case (2.2), easy to know
that $\delta$ is impossible to be $4$ and hence $>4$.

Case (3) If $i=u$ but $j\neq v$, the two edges can be rewritten to
$(x_{i}^{a_{1}+\frac{1}{2}},x_{j}^{a_{1}+\frac{1}{2}},x_{v}^{a_{2}}%
)\cup(X_{others}^{A}=%
\mathbb{N}
\backslash a_{1},a_{1}+1)$ and $(x_{i}^{b_{1}+\frac{1}{2}},x_{j}^{b_{2}}%
,x_{v}^{b_{1}+\frac{1}{2}})\cup(X_{others}^{B}=%
\mathbb{N}
\backslash b_{1},b_{1}+1)$. Because $\delta(x_{j})\geq\frac{1}{2},\delta
(x_{v})\geq\frac{1}{2}$, then $\delta\geq1$.

Case (3.1) To reach $\delta=1$, it must be $\delta(x_{i})=\delta
(x_{others})=0$, $\delta(x_{j})=\delta(x_{v})=\frac{1}{2}$. Therefore we have
solutions that $a_{1}=b_{1}$, $%
\genfrac{\{}{.}{0pt}{}{a_{1}=b_{2}}{a_{1}=b_{2}-1}%
$, and $%
\genfrac{\{}{.}{0pt}{}{a_{2}=b_{1}}{a_{2}=b_{1}+1}%
$. Easy to check these solutions will lead to such as $b_{2}=b_{1}$,
$b_{2}=b_{1}+1$, $a_{2}=a_{1}$, $a_{2}=a_{1}+1$ -- all are impossible since
they have been excluded from their codes. As a result, the next minimum
$\delta=2$.

Case (3.2) To reach $\delta=2$, there two possible cases of chromatic distance
at each component.

Case (3.2.1) $\delta(x_{i})=1$, $\delta(x_{j})=\delta(x_{v})=\frac{1}{2}$, and
$\delta(X_{others})=0$.

This case leads to solutions that $%
\genfrac{\{}{.}{0pt}{}{a_{1}=b_{1}+1}{a_{1}=b_{1}-1}
$, $%
\genfrac{\{}{.}{0pt}{}{a_{1}=b_{2}}{a_{1}=b_{2}-1}%
$, and $%
\genfrac{\{}{.}{0pt}{}{a_{2}=b_{1}}{a_{2}=b_{1}+1}%
$, which correspond to 8 combinations but only 4 of them will lead to the
allowed edges $(x_{i}^{b_{1}-\frac{1}{2}},x_{j}^{b_{1}-\frac{1}{2}}%
,x_{v}^{b_{1}+1})$ and $(x_{i}^{b_{1}+\frac{1}{2}},x_{j}^{b_{1}-1}%
,x_{v}^{b_{1}+\frac{1}{2}})$. Easy to check that the two edges are joint to a
3-I vertex $(x_{i}^{\frac{b_{1}}{3}},x_{j}^{\frac{b_{1}}{3}},x_{v}%
^{\frac{b_{1}}{3}})$, and in this case $\gamma=3$, $\eta_{1}\ncong\eta_{2}$.

Case (3.2.2) $\delta(x_{i})=0,\delta(x_{j})=\frac{3}{2},\delta(x_{v})=\frac
{1}{2},$and $\delta(x_{others})=0$. These equations give solutions
$a_{1}=b_{1}$, $%
\genfrac{\{}{.}{0pt}{}{a_{1}=b_{2}+1}{a_{1}=b_{2}-2}%
$, and $%
\genfrac{\{}{.}{0pt}{}{a_{2}=b_{1}}{a_{2}=b_{1}+1}%
$. Easy to check these solutions are not allowed.

Case (3.3) To reach $\delta=3$, there 4 possible combinations of chromatic
distance at each component $(\delta(x_{i})$, $\delta(x_{j})$, $\delta(x_{v})$,
$\delta(x_{other}))$: $(1,\frac{1}{2},\frac{1}{2},1)$, $(2,\frac{1}{2}%
,\frac{1}{2},0)$, $(0,\frac{1}{2},\frac{5}{2},0)$, and $(0,\frac{3}{2}%
,\frac{3}{2},0)$. Easy to check that the first three combinations are
impossible and only the fourth are allowed to give edges such as
$(x_{i}^{a_{1}+\frac{1}{2}},x_{j}^{a_{1}+\frac{1}{2}},x_{v}^{a_{1}-1})$ and
$(x_{i}^{a_{1}+\frac{1}{2}},x_{j}^{a_{1}-1},x_{v}^{a_{1}+\frac{1}{2}})$, which
also gives $\gamma=2$ and $\eta_{1}\cong\eta_{2}$. Using $E2V(\eta,3I)$ we
know the two edges are both joint to a 3-I vertex $(x_{i}^{\frac{a_{1}}{3}%
},x_{j}^{\frac{a_{1}}{3}},x_{v}^{\frac{a_{1}}{3}})$.

Case (3.4) To reach $\delta=4$, there 8 possible combinations of chromatic
distance at each component $(\delta(x_{i})$, $\delta(x_{j})$, $\delta(x_{v})$,
$\delta(x_{other}))$: $(1,\frac{1}{2},\frac{1}{2},2)$, $(2,\frac{1}{2}%
,\frac{1}{2},1)$, $(3,\frac{1}{2},\frac{1}{2},0)$, $(1,\frac{3}{2},\frac{3}%
{2},0)$, $(0,\frac{3}{2},\frac{3}{2},1)$, $(1,\frac{1}{2},\frac{5}{2},0)$,
$(0,\frac{1}{2},\frac{5}{2},1)$, and $(0,\frac{1}{2},\frac{7}{2},0)$. Checking
these combinations we know that only the first combination is allowed, which
gives the edge codes $(x_{i}^{a_{1}+\frac{1}{2}},x_{j}^{a_{1}+\frac{1}{2}%
},x_{v}^{a_{1}-1})\cup(x_{g}^{a_{2}},x_{h}^{a_{2}+1})$ and $(x_{i}%
^{a_{1}-\frac{1}{2}},x_{j}^{a_{1}+1},x_{v}^{a_{1}-\frac{1}{2}})\cup
(x_{g}^{a_{2}+1},x_{h}^{a_{2}})$, but they are not joint with the same 3-I
vertex and also $\gamma=5$.

Summarizing the above cases (1), (2.1.2), (2.2.2), (3.2.1), and (3.3) we know
that the conditions in Table \ref{Table1} are also sufficient for reasoning
E-E joint relations.
\end{proof}

Note that if we do not care the types of the joint vertex, then we can
integrate conditions in Table 1 to%
\begin{equation}
\cap(\eta_{1},\eta_{2})=\varphi\Leftrightarrow\left\{
\begin{array}
[c]{l}%
\delta\leq3\\
(\delta,\gamma)=(4,3),\eta_{1}\ncong\eta_{2}%
\end{array}
\right.
\end{equation}

\subsubsection{Edge-Cell (E-C) relations}

There are three types of relations between an edge and a cell: (1)
\emph{contain}: the edge is one boundary of the cell, i.e., $\cap(\zeta
,\eta)=\eta$ (Fig.6j), (2) \emph{joint}: the cell only share a vertex with the
edge, i.e., $\cap(\zeta,\eta)=\varphi$ (Fig.6k), and (3) \emph{disjoint}: they
do not share any particles, i.e., $\cap(\zeta,\eta)=\varnothing$ (Fig.6l).

\begin{proposition}
\label{Proposition_EC_contain}Given a cell $\zeta$ and a edge $\eta$, then%
\begin{equation}
\cap(\zeta,\eta)=\eta\Leftrightarrow\delta(\zeta,\eta)=1
\end{equation}

\end{proposition}

\begin{proof}
A cell $\zeta$ could be taken as the space closed by edges that are from
either 2-I or 3-I units, and this gives $\zeta\cap\eta=\eta\Rightarrow$
$\delta(\zeta,\eta)=1$.

Suppose $\zeta$ bears a code $(x_{i}^{a_{1}},x_{j}^{a_{2}})\cup(X_{others}%
^{A})$ and $\eta$ bears a code $(x_{i}^{b_{1}+\frac{1}{2}},x_{j}^{b_{1}%
+\frac{1}{2}})\cup(X_{others}^{B})$. To make $\delta(\zeta,\eta)=1$, it must
be that $|x_{i}^{a_{1}}-x_{i}^{b_{1}+\frac{1}{2}}|=\frac{1}{2}$,
$|x_{j}^{a_{2}}-x_{j}^{b_{1}+\frac{1}{2}}|=\frac{1}{2}$, and $|X_{others}%
^{A}-X_{others}^{B}|=0$, so the solutions are $%
\genfrac{\{}{.}{0pt}{}{a_{1}=b_{1}}{a_{1}=b_{1}-1}%
$ and $%
\genfrac{\{}{.}{0pt}{}{a_{2}=b_{1}}{a_{2}=b_{1}-1}%
$. Because $a_{1}\neq a_{2}$, then we have only two solutions $%
\genfrac{\{}{.}{0pt}{}{a_{1}=b_{1},a_{2}=b_{1}-1}{a_{2}=b_{1},a_{1}=b_{1}-1}%
$. It is easy to know that the two solutions are just the two cells who share
the edge, that is, $\zeta\cap\eta=\eta$.
\end{proof}

Similar to the function $E2V$, a function $C2E$ is used for calculating all
edges that bound a cell.

\begin{notation}
The procedure of $C2E(\zeta)$:

Let $\zeta$ is an cell with base $%
\mathbb{N}
\lbrack0,n-1]$. (1) Find the minimum component $z$ in its code, assume it is
$x_{i}^{z}$; (2) Find $z+1$, assume it is $x_{j}^{z+1}$, then change
$x_{i}^{z}$ and $x_{j}^{z+1}$ to $x_{i}^{z+\frac{1}{2}}$ and $x_{j}%
^{z+\frac{1}{2}}$, respectively. (3) Let $z=z+1$ and repeat (1) and (2) until
$z=n-2$.
\end{notation}

Therefore in theory each cell should be bounded by $n-1$ edges, i.e., they are
$n$-hedras, but in fact many cells are triangles, quadrilaterals, or polygons
with edges much less then $n-1$, indicating that a large number of edges do
not emerge in plane.

\begin{proposition}
\label{Proposition_EC_joint}Given a cell $\zeta$ and a edge $\eta$ in $OACD(n,%
\mathbb{R}
^{n-1})$, then%
\begin{equation}
\cap(\zeta,\eta)=\varphi\Leftrightarrow3\leq\delta(\zeta,\eta)\leq4
\end{equation}

\end{proposition}

\begin{proof}
Table \ref{Table1} shows that in 2-I and 3-I units, if a cell and an edge are
joint, then their chromatic distances are either 3 or 4.

Suppose $\zeta$ bears a code $(x_{i}^{a_{1}},x_{j}^{a_{2}})\cup(X_{others}%
^{A})$ and $\eta$ bears a code $(x_{i}^{b_{1}+\frac{1}{2}},x_{j}^{b_{1}%
+\frac{1}{2}})\cup(X_{others}^{B}\backslash X_{removed}^{b,b_{1}+1})$.

Case (1) To make $\delta(\zeta,\eta)=3$, all possible combinations of
$(\delta(x_{i})$, $\delta(x_{j})$, $\delta(x_{others}))$ are $(\frac{1}%
{2},\frac{1}{2},2)$, $(\frac{1}{2},\frac{3}{2},1)$, $(\frac{1}{2},\frac{5}%
{2},0)$, and $(\frac{3}{2},\frac{3}{2},0)$. Checking these combinations we
know that only the cases $(\frac{1}{2},\frac{1}{2},2)$ and $(\frac{1}{2}%
,\frac{3}{2},1)$ are allowed, in which $(\frac{1}{2},\frac{1}{2},2)$ gives two
solutions $%
\genfrac{\{}{.}{0pt}{}{a_{1}=b_{1}}{a_{2}=b_{1}+1}%
$ and $%
\genfrac{\{}{.}{0pt}{}{a_{1}=b_{1}+1}{a_{2}=b_{1}}%
$, namely, two cells are opposite to an edge in a 2-I unit, and $(\frac{1}%
{2},\frac{3}{2},1)$ gives two solutions $%
\genfrac{\{}{.}{0pt}{}{a_{1}=b_{1}}{a_{2}=b_{1}+2}%
$ and $%
\genfrac{\{}{.}{0pt}{}{a_{1}=b_{1}+1}{a_{2}=b_{1}-1}%
$, namely, two edges are interval to a cell in a 3-I unit. And easy to know
that their code distances are 4 and 3, respectively.

Case (2) To make $\delta(\zeta,\eta)=4$, all possible combinations of
$(\delta(x_{i})$, $\delta(x_{j})$, $\delta(x_{others}))$ are $(\frac{1}%
{2},\frac{1}{2},3)$, $(\frac{1}{2},\frac{3}{2},2)$, $(\frac{1}{2},\frac{5}%
{2},1)$, $(\frac{3}{2},\frac{3}{2},1)$, $(\frac{3}{2},\frac{5}{2},0)$, and
$(\frac{1}{2},\frac{7}{2},0)$. Similarly, only $(\frac{1}{2},\frac{3}{2},2)$
is allowed to give two solutions $%
\genfrac{\{}{.}{0pt}{}{a_{1}=b_{1}}{a_{2}=b_{1}-1}%
$ and $%
\genfrac{\{}{.}{0pt}{}{a_{1}=b_{1}+1}{a_{2}=b_{1}+2}%
$ that correspond to the two cells being opposite to an edge in a 3-I unit,
with $\gamma=3$.
\end{proof}

Based on the above two propositions, the disjoint E-C relation can be
determined by the below corollary.

\begin{corollary}
Given a cell $\zeta$ and a edge $\eta$, then%
\begin{equation}
\cap(\zeta,\eta)=\varnothing\Leftrightarrow\delta(\zeta,\eta)>4
\end{equation}

\end{corollary}

We can use $C2E$ and then $E2V$ function, that is, $E2V(C2E(\zeta))$, to
obtain all vertexes contained by a cell. Also we can define a new function
$C2V(\zeta)$ to directly find out all of such that vertexes, using the similar
procedures in $E2V(\eta)$. Therefore using functions $C2E$ and $C2V$, E-C
relations can be also expressed by
\begin{align}
\cap(\zeta,\eta)  & =\eta\Leftrightarrow\eta\in C2E(\zeta)\\
\cap(\zeta,\eta)  & =\varphi\Leftrightarrow\eta\notin C2E(\zeta)\wedge
C2V(\zeta)\cap E2V(\eta)\neq\varnothing\nonumber\\
\cap(\zeta,\eta)  & =\varnothing\Leftrightarrow C2V(\zeta)\cap E2V(\eta
)=\varnothing\nonumber
\end{align}

Note that because some joint vertexes may be hidden in high dimensional
spaces, therefore some joint E-C relations may appear to be disjoint in $%
\mathbb{R}
^{2}$ plane.

\subsubsection{Cell-Cell (C-C) relations}

The C-C relations can also be three types: (1) \emph{connected}: two cells
share a common edge, i.e., $\cap(\zeta_{1},\zeta_{2})=\eta$ (Fig.6m), (2)
\emph{joint}: they share a common vertex, i.e., $\cap(\zeta_{1},\zeta
_{2})=\varphi$ (Fig.6n), and (3) \emph{disjoint}: they do not share any
particles, i.e., $\cap(\zeta_{1},\zeta_{2})=\varnothing$ (Fig.6o). The
connected and joint relations actually correspond to the five C-C relations
occurred in 2-I or 3-I units (Table 1), where the adjacent corresponds to the
connected, and the opposite and interval both correspond to the joint.

\begin{proposition}
\label{Proposition_CC_touch}Given two cells $\zeta_{1}$ and $\zeta_{2}$,%
\begin{equation}
\cap(\zeta_{1},\zeta_{2})=\eta\Leftrightarrow\delta(\zeta_{1},\zeta_{2})=2
\end{equation}

\end{proposition}

\begin{proof}
We have proved that if two cell are adjacent, then their chromatic distance is
2. Now we need to prove if $\delta(\zeta_{1},\zeta_{2})=2$, then they must be
adjacent. Suppose $\zeta_{1}$ and $\zeta_{2}$ are with codes $(x_{i}^{a_{1}%
},x_{j}^{a_{2}})\cup(X_{others}^{A})$ and $(x_{i}^{b_{1}},x_{j}^{b_{2}}%
)\cup(X_{others}^{B})$. The only way giving $\delta=2$ is that $|a_{1}%
-b_{1}|=1$, $|a_{2}-b_{2}|=1$, and $|X_{others}^{A}-X_{others}^{B}|=0$. This
leads to solutions $%
\genfrac{\{}{.}{0pt}{}{a_{1}=b_{1}}{a_{1}=b_{2}}%
$, $%
\genfrac{\{}{.}{0pt}{}{a_{2}=b_{1}}{a_{2}=b_{2}}%
$, $%
\genfrac{\{}{.}{0pt}{}{a_{1}=b_{2}+1}{a_{1}=b_{1}-1}%
$ and $%
\genfrac{\{}{.}{0pt}{}{a_{2}=b_{2}+1}{a_{2}=b_{2}-1}%
$, then we can obtain the allowed two solutions that (1) $a_{1}=b_{2}$,
$a_{2}=b_{1}$, $a_{1}=b_{1}-1$, $a_{2}=b_{2}+1$; and (2) $a_{1}=b_{2}$,
$a_{2}=b_{1}$, $a_{1}=b_{1}-1$, $a_{2}=b_{2}-1$. Substituting the two
solutions back to chromatic codes of $\zeta_{1}$ and $\zeta_{2}$, then they
turn to such as $\zeta_{1}=(x_{i}^{a_{1}},x_{j}^{a_{1}+1})$, $\zeta_{2}%
=(x_{i}^{a_{1}+1},x_{j}^{a_{1}})$, or $\zeta_{1}=(x_{i}^{b_{1}-1},x_{j}%
^{b_{1}})$, $\zeta_{2}=(x_{i}^{b_{1}},x_{j}^{b_{1}-1})$, or similar codes.
Using $C2E(\zeta)$ we know that they both joint with an edge $(x_{i}%
^{a_{1}+\frac{1}{2}},x_{j}^{a_{1}+\frac{1}{2}})$ or $(x_{i}^{b_{1}-\frac{1}%
{2}},x_{j}^{b_{1}-\frac{1}{2}})$, or similar codes.
\end{proof}

\begin{proposition}
\label{Proposition_CC_joint}Given two cells $\zeta_{1}$ and $\zeta_{2}$,%
\begin{equation}
\cap(\zeta_{1},\zeta_{2})=\varphi\Leftrightarrow\delta(\zeta_{1},\zeta_{2})=4
\end{equation}

\end{proposition}

\begin{proof}
The properties \ref{P_2I3IUnits_Distance} shows that if two cells are joint to
a vertex, then their chromatic distance is 4. Now let us explore all possible
cases that make chromatic distance is 4.

First we rewrite $\zeta_{1}\ $and $\zeta_{2}$ to $(X_{diff_{1}}^{D_{1}}%
)\cup(X_{same}^{S_{1}})$ and $(X_{diff_{2}}^{D_{2}})\cup(X_{same}^{S_{2}})$,
where $X_{same}^{S_{1}}=X_{same}^{S_{2}}$. This indicates that $\delta
(\zeta_{1},\zeta_{2})$ is only contributed by $X_{diff}$, where for each
component, $x_{diff_{1}}^{d_{1}}\neq x_{diff_{2}}^{d_{2}}$. Suppose $X_{diff}
$ contains $m$ components, then

Case (1) $m=2$, namely, $\gamma(\zeta_{1},\zeta_{2})=2$.

Let $X_{diff_{1}}^{D_{1}}=$ $(x_{i}^{a_{1}},x_{j}^{a_{2}})$ and $X_{diff_{2}%
}^{D_{2}}=$ $(x_{i}^{b_{1}},x_{j}^{b_{2}})$, then we have $a_{1}\neq b_{1}$,
$a_{2}\neq b_{2}$. Also all cells are equi-base, $\zeta_{1}\cong\zeta_{2}$,
and $X_{same}^{S_{1}}=X_{same}^{S_{2}}\Rightarrow S_{1}\cong S_{2}$, therefore
we always have $X_{diff_{1}}^{D_{1}}\cong X_{diff_{2}}^{D_{2}}\Rightarrow%
\genfrac{\{}{.}{0pt}{}{a_{1}=b_{1}}{a_{1}=b_{2}}%
$, $%
\genfrac{\{}{.}{0pt}{}{a_{2}=b_{1}}{a_{2}=b_{2}}%
$. Because $x_{diff_{1}}^{d_{1}}\neq x_{diff_{2}}^{d_{2}}$, therefore we have
$a_{1}=b_{2}$ and $a_{2}=b_{1}$, and then $\delta=2|a_{1}-a_{2}|=4\Rightarrow
a_{1}=a_{2}+2$ or $a_{1}=a_{2}-2$. We thus further rewrite $X_{diff_{1}%
}^{D_{1}}$ and $X_{diff_{2}}^{D_{2}}$ to $(x_{i}^{a_{2}+2},x_{j}^{a_{2}}%
)\cup(x_{k}^{a_{2}+1})$ and $(x_{i}^{a_{2}},x_{j}^{a_{2}+2})\cup(x_{k}%
^{a_{2}+1}) $, because $a_{2}+1$ must be somewhere in $X_{same}$, assuming it
is $k$. Then $\zeta_{1}$ should be bounded by two edges $\eta_{_{11}}%
=(x_{i}^{a_{2}+2}$, $x_{j}^{a_{2}+\frac{1}{2}}$, $x_{k}^{a_{2}+\frac{1}{2}})$,
$\eta_{_{12}}=(x_{i}^{a_{2}+\frac{3}{2}}$, $x_{j}^{a_{2}}$, $x_{k}%
^{a_{2}+\frac{3}{2}})$, and $\zeta_{2}$ should be bounded by two edges
$\eta_{_{21}}=(x_{i}^{a_{2}+\frac{1}{2}}$, $x_{j}^{a_{2}+2}$, $x_{k}%
^{a_{2}+\frac{1}{2}})$, $\eta_{_{22}}=(x_{i}^{a_{2}}$, $x_{j}^{a_{2}+\frac
{3}{2}}$, $x_{k}^{a_{2}+\frac{3}{2}})$. From the properties of 3-I unit we
know that $\eta_{_{11}}$ and $\eta_{_{22}}$ are joint to a $\varphi^{3I}$ and
both in bisector $pb\left\langle j,k\right\rangle $, and in the same way,
$\eta_{_{12}}$ and $\eta_{_{21}}$ are joint to the same $\varphi^{3I}$ and
also both in bisector $pb\left\langle i,k\right\rangle $ -- this is just the
case that the two cells are in a 3-I units and with an opposite relation.

Case (2) $m=3$, namely, $\gamma(\zeta_{1},\zeta_{2})=3$.

We rewrite $X_{diff_{1}}^{D_{1}}$ and $X_{diff_{2}}^{D_{2}}$ to $(x_{i}%
^{a_{1}},x_{j}^{a_{2}},x_{k}^{a_{3}})$ and $(x_{i}^{b_{1}},x_{j}^{b_{2}}%
,x_{k}^{b_{3}})$. Because $X_{diff_{1}}^{D_{1}}\cong X_{diff_{2}}^{D_{2}}$ and
$x_{diff_{1}}^{d_{1}}\neq x_{diff_{2}}^{d_{2}}$, then we have $a_{1}=b_{3}$,
$a_{2}=b_{1}$, $a_{3}=b_{2}$, or $a_{1}=b_{2}$, $a_{2}=b_{3}$, $a_{3}=b_{1}$.
Suppose $a_{2}=a_{1}+\Delta_{1}$ and $a_{3}=a_{1}+\Delta_{1}+\Delta_{2}$, with
$\Delta_{1}\geq1$ and $\Delta_{2}\geq1$, then easy to obtain that
$\delta=2\Delta_{1}+2\Delta_{2}$. Therefore $\delta=4$, only if $\Delta_{1}=$
$\Delta_{2}=1$. The $X_{diff_{1}}^{D_{1}}=(x_{i}^{a_{1}},x_{j}^{a_{1}+1}%
,x_{k}^{a_{1}+2})$ and $X_{diff_{2}}^{D_{21}}=(x_{i}^{a_{1}+1},x_{j}^{a_{1}%
+2},x_{k}^{a_{1}})$ or $X_{diff_{2}}^{D_{22}}=(x_{i}^{a_{1}+2},x_{j}^{a_{1}%
},x_{k}^{a_{1}+1})$, or their corresponding permutations. Using $C2E(\zeta)$,
we can get $\zeta_{1} $ has two edges $\eta_{11}=$ $(x_{i}^{a_{1}+\frac{1}{2}%
},x_{j}^{a_{1}+\frac{1}{2}})$, $\eta_{12}=$ $(x_{j}^{a_{1}+\frac{3}{2}}%
,x_{k}^{a_{1}+\frac{3}{2}})$, and for case $X_{diff_{1}}^{D_{21}}$,
$\zeta_{21}$ has two edges $\eta_{211}=$ $(x_{i}^{a_{1}+\frac{3}{2}}%
,x_{j}^{a_{1}+\frac{3}{2}})$, $\eta_{212}=$ $(x_{i}^{a_{1}+\frac{1}{2}}%
,x_{k}^{a_{1}+\frac{1}{2}})$; and for case $X_{diff_{2}}^{D_{22}}$,
$\zeta_{22}$ has two edges $\eta_{221}=$ $(x_{j}^{a_{1}+\frac{1}{2}}%
,x_{k}^{a_{1}+\frac{1}{2}})$, $\eta_{222}=$ $(x_{i}^{a_{1}+\frac{3}{2}}%
,x_{k}^{a_{1}+\frac{3}{2}})$; Using $E2V(\eta)$ we know that all the six edges
are joint to the same $\varphi^{3I}=(x_{i}^{a_{1}+1},x_{j}^{a_{1}+1}%
,x_{k}^{a_{1}+1})$. Also according to the Proposition
\ref{Proposition_EE_collinear}, the edges in pairs $(\eta_{11},\eta_{211})$,
$(\eta_{12},\eta_{221})$, $(\eta_{212},\eta_{222})$ are both in the same
bisector. These cell-edge-vertex relations are just the case that two cells
$\zeta_{21}$ and $\zeta_{22}$ are in interval relations to the cell $\zeta
_{1}$.

Case (3) $m=4$, namely, $\gamma(\zeta_{1},\zeta_{2})=4$.

We rewrite $X_{diff_{1}}^{D_{1}}$ and $X_{diff_{2}}^{D_{2}}$ to $(x_{i}%
^{a_{1}},x_{j}^{a_{2}},x_{u}^{a_{3}},x_{v}^{a_{4}})$ and $(x_{i}^{b_{1}}%
,x_{j}^{b_{2}},x_{u}^{b_{3}},x_{v}^{b_{4}})$. Because each component in
$X_{diff}$ at least contributes $1$ chromatic distance to $\delta$, therefore
to make $\delta=4$, it must be $\delta(x_{i})=\delta(x_{j})=\delta
(x_{u})=\delta(x_{v})=1$. Then the solutions are $%
\genfrac{\{}{.}{0pt}{}{a_{1}=b_{1}+1}{a_{1}=b_{1}-1}%
$, $%
\genfrac{\{}{.}{0pt}{}{a_{2}=b_{2}+1}{a_{2}=b_{2}-1}%
$, $%
\genfrac{\{}{.}{0pt}{}{a_{3}=b_{3}+1}{a_{3}=b_{3}-1}%
$, $%
\genfrac{\{}{.}{0pt}{}{a_{4}=b_{4}+1}{a_{4}=b_{4}-1}%
$, and also it is needed that $X_{diff_{1}}^{D_{1}}\cong X_{diff_{2}}^{D_{2}}%
$. Assume $a_{2}=a_{1}+\Delta_{1}$, $a_{3}=a_{1}+\Delta_{1}+\Delta_{2}$, and
$a_{4}=a_{1}+\Delta_{1}+\Delta_{2}+\Delta_{3}$, then the only solution for
$X_{diff_{2}}^{D_{2}}$ is ($a_{1}+\Delta_{1}$, $a_{1}$, $a_{1}+\Delta
_{1}+\Delta_{2}+\Delta_{3}$, $a_{1}+\Delta_{1}+\Delta_{2})$ with $\Delta
_{1}=\Delta_{3}=1$, corresponding to the opposite cell in 2-I unit.

Case (4) $m\geq5$.

Because each components in $X_{diff}$ at least contribute $1$ chromatic
distance to $\delta$, therefore $\delta\geq5$.
\end{proof}

Another approach to reason C-C relations is using $C2E(\zeta)$ function to
calculate all edges of two cells, and if among them two edges are equal, then
the two cells must be connected. We can also use $C2V(\zeta)$ to calculate all
vertexes contained by two cells and hence determine if they are joint with a
vertex, but here we will encounter the same problem that the joint vertex is hidden.

\subsection{Spatial topology between complexes}

Real objects or geographical entities in space usually occupy massive spatial
particles. Topological relations and computations among complexes are hence
much more complicated than those among particles. As the union set of
particles, a complex may contain different types of particles, for example,
containing two cells, two edges, and three vertexes, such complexes are called
\emph{mixed complexes}; or it contains only a single type of particles, for
example, containing only vertexes, edges, or cells, such complexes are called
\emph{uniform complexes}. In addition, there are also two information
scenarios: for a given complex, (1) we know its code as well as its all
elemental particles, and (2) we only know its code but do not know its
elemental particles. This section demonstrates an tentative study of spatial
complex topology, particularly focusing on the most important uniform complex
-- cluster, as well as the scenario (1) that we know each elemental cell.

\subsubsection{Spatial connectivity of a cluster}

The spatial connectivity is an important issue for analyzing complexes and
clusters. In a general sense, the connectivity of clusters in OACD and SCM is
similar to those in graph theory, complex network, algebraic geometry, and
point set topology. A disconnected cluster is usually treated as a number of
connected clusters rather than a single cluster.

Let us define the connectivity of a cluster. If a cluster $\xi$ contains two
cells $\zeta_{1}$ and $\zeta_{2}$ and they are connected as in Proposition
\ref{Proposition_CC_touch}, namely, $\delta(\zeta_{1},\zeta_{2})=2$, then
there is a \emph{path} linking them, denoted by $\rho(\zeta_{1},\zeta_{2})$.
If a cluster contains three cells $\zeta_{1}$, $\zeta_{2}$, and $\zeta_{3}%
,$and there is a path $\rho(\zeta_{1},\zeta_{2})$, and another path
$\rho(\zeta_{2},\zeta_{3})$, then we define a path $\rho(\zeta_{1},\zeta_{3})$
between $\zeta_{1}$ and $\zeta_{3}$, and call them \emph{path-connected} by a
\emph{path-cell} $\zeta_{2}$, that is, $\rho(\zeta_{1},\zeta_{3})=(\zeta_{2}%
)$. Similarly, any two cells are path-connected if they are linked by a series
of path-cells.

\begin{definition}
Given a cluster $\xi\{\zeta_{1},\zeta_{2},\ldots,\zeta_{n}\}$, it is
\emph{connected} if it meets two conditions: (1) any two cells are
path-connected, and (2) all path-cells are the elements of the cluster.
\end{definition}

\begin{notation}
The function $Conn(\xi)$ returns the connectivity of $\xi$. It can be carried
out by steps (1) select any one cell from $\xi$ as the seed of the connected
set $\boldsymbol{C}_{c}$ and the other cells remain as the waiting-list set
$\boldsymbol{C}_{w}$, (2) search cells in $\boldsymbol{C}_{w}$ to find out the
cell $\zeta_{w}$ which is connected to any cell $\zeta_{c}$ in $\boldsymbol{C}%
_{c}$, that is,\thinspace$\delta(\zeta_{c},\zeta_{w})=2$. (3) If found, then
move $\zeta_{w}$ from $\boldsymbol{C}_{w}$ to $\boldsymbol{C}_{c}$, and repeat
step (2) until $\boldsymbol{C}_{w} $ becomes empty, and then return
$Conn(\xi)=1$, meaning $\xi$ is connected; If not found, return $Conn(\xi)=0$,
meaning $\xi$ is disconnected.
\end{notation}

\subsubsection{Types and reasoning of cluster-cluster topological relations}

Given two clusters $\xi_{1}$ and $\xi_{2}$, their cluster-cluster (Cs-Cs)
topological relations are demonstrated in Fig.7, such as equal, contain,
touch, and overlap. Because clusters are union set of cells and if their
elemental cells are known, say, $\xi_{1}=\{\zeta_{11},\zeta_{12},\cdots
,\zeta_{1n}\}$ and $\xi_{2}=\{\zeta_{21},\zeta_{22},\cdots,\zeta_{2m}\} $,
then some of Cs-Cs relations are easy to determine by using below set
operations.%
\begin{align}
\xi_{1}\text{ equals }\xi_{2}  & \Leftrightarrow\xi_{1}\cap\xi_{2}=\xi_{1}%
=\xi_{2}\\
\xi_{1}\text{ contains }\xi_{2}  & \Leftrightarrow\xi_{1}\supset\xi_{2}\text{
}\nonumber\\
\xi_{1}\text{ disjoints }\xi_{2}  & \Leftrightarrow\xi_{1}\cap\xi
_{2}=\varnothing\nonumber\\
\xi_{1}\text{ overlaps }\xi_{2}  & \Leftrightarrow\xi_{1}\cap\xi_{2}%
\neq\varnothing\neq\xi_{1}\neq\xi_{2}\nonumber
\end{align}

\begin{figure*}
\centerline{\includegraphics[width=4in]{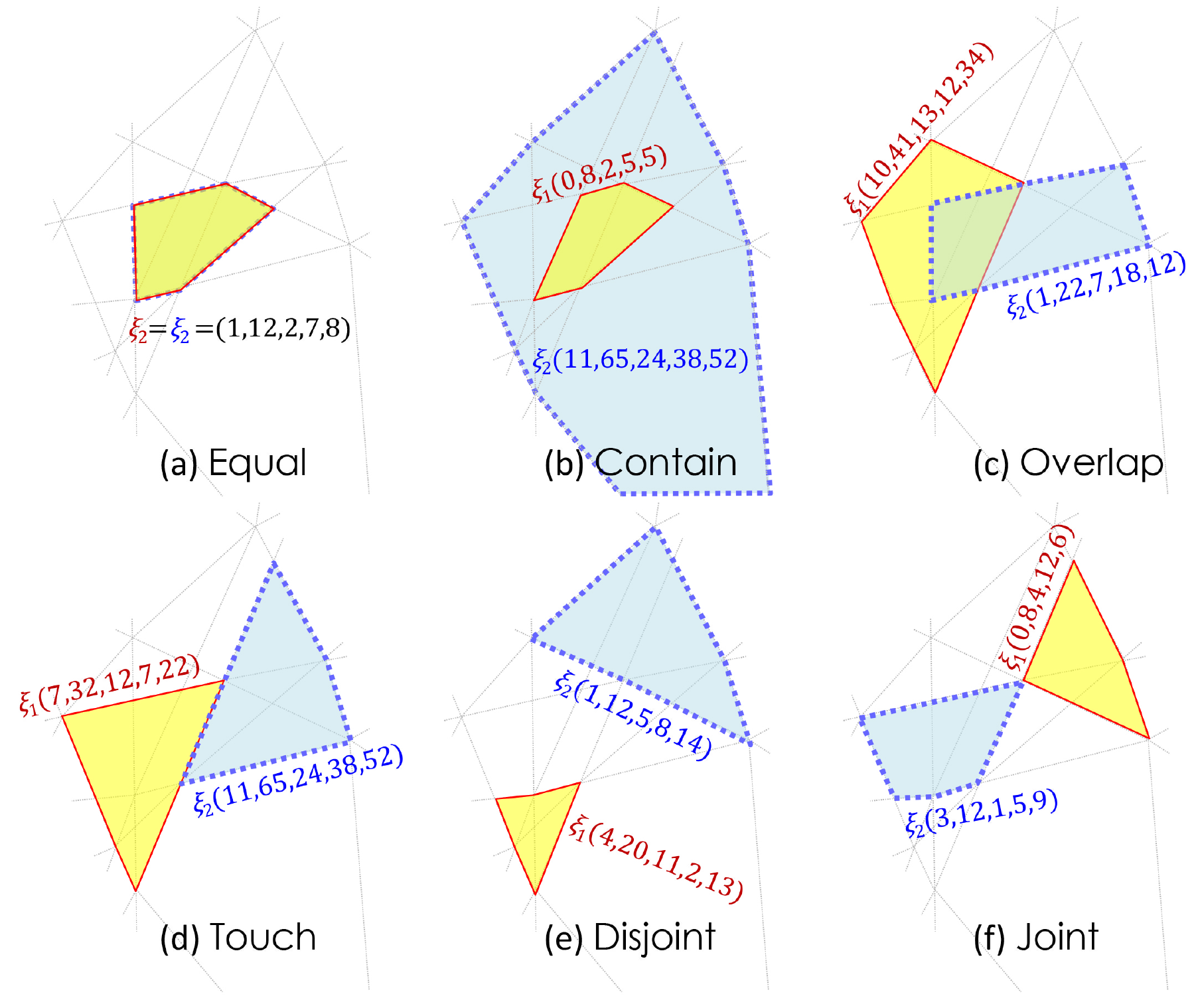}}
\caption{Six types of complex topological relations in full-OACD.}
\label{Fig7}
\end{figure*}

Another usual topological relation between two clusters is adjacency (Fig.7d),
which can be determined by%
\begin{equation}
\xi_{1}\text{ touch }\xi_{2}\Leftrightarrow\xi_{1}\cap\xi_{2}=\varnothing
\wedge Conn(\xi_{1}\cup\xi_{2})=1
\end{equation}
Or we can use $C2E$ and $C2V$ function to compare their edges and vertexes,
that is,%
\begin{align}
\xi_{1}\text{ touch }\xi_{2}  & \Leftrightarrow\xi_{1}\cap\xi_{2}%
=\varnothing\wedge C2E(\xi_{1})\cap C2E(\xi_{2})\neq\varnothing\\
\xi_{1}\text{ joint }\xi_{2}  & \Leftrightarrow C2E(\xi_{1})\cap C2E(\xi
_{2})=\varnothing\wedge C2V(\xi_{1})\cap C2V(\xi_{2})\neq\varnothing\nonumber
\end{align}
where $C2E(\xi)=\underset{\zeta\in\xi}{\cup}C2E(\zeta)$ and $C2V(\xi
)=\underset{\zeta\in\xi}{\cup}C2V(\zeta)$.

A more comprehensive but perhaps more complicated method to explore Cs-Cs
relations is examining all C-C relations among their elemental cells. Given
two complexes $\Theta_{1}(\Omega_{11},\Omega_{12},\ldots,\Omega_{1n})$ and
$\Theta_{2}(\Omega_{21},\Omega_{22},\ldots,\Omega_{2m})$, their
\emph{chromatic-distance matrix} is defined by $dM(\Theta_{1},\Theta
_{2})=[\delta_{ij}]_{n\times m}$, where $\delta_{ij}=\delta(\Omega_{1i}%
,\Omega_{2j})$, that is,
\begin{equation}
dM(\Theta_{1},\Theta_{2})=\left[
\begin{array}
[c]{cccc}%
\delta(\Omega_{11},\Omega_{21}) & \delta(\Omega_{11},\Omega_{22}) & \cdots &
\delta(\Omega_{11},\Omega_{2m})\\
\delta(\Omega_{12},\Omega_{21}) & \delta(\Omega_{12},\Omega_{22}) & \cdots &
\delta(\Omega_{12},\Omega_{2m})\\
\vdots & \vdots & \ddots & \vdots\\
\delta(\Omega_{1n},\Omega_{21}) & \delta(\Omega_{1n},\Omega_{22}) & \cdots &
\delta(\Omega_{1n},\Omega_{2m})
\end{array}
\right]
\end{equation}

Given a complex $\Theta$, $dM(\Theta,\Theta)$ is also called the
\emph{internal matrix} of $\Theta$, and denoted by $iM(\Theta)$, which can be
used to determine the connectivity of a cluster. Replacing all $\delta=2$ in
$iM(\xi)$ by 1 and all others by 0, then we get an adjacency matrix $aM(\xi)$,
the same one used in graph theory. The $aM(\xi)$ can be transferred to a
reachability matrix $rM(\xi)=$ $aM(\xi)+$ $aM(\xi)^{2}+\cdots+aM(\xi)^{n}$, or
by such as Floyd-Warshall, Thorup, or Kameda's algorithms \cite{Cormen2001},
\cite{Thorup2004}, \cite{Kameda1975}. If $rM(\xi)=1$, then $\xi$ is an
connected cluster.

By using Proposition \ref{Proposition_CC_touch} and \ref{Proposition_CC_joint}%
, we can determine Cs-Cs relations by whether some particular chromatic
distances are found in $dM$. For example, if we found 0 or 2 in $dM$, then it
means a cell in one cluster is equal or connected to a cell in the other cluster.

\begin{notation}
Function $cdn(dM,k)$ returns the number of $\delta(\Omega_{1},\Omega_{2})=k$
in a chromatic-distance matrix $dM(\Theta_{1},\Theta_{2})$. $k$ also can be
some conditions such as $>0$, or $\neq2$.
\end{notation}

We then can determine Cs-Cs relations by using $dM$, $cdn$ function, and the
below rules.%
\begin{align}
\xi_{1}\text{ equals }\xi_{2}  & \Leftrightarrow|\xi_{2}|=\tfrac{1}%
{2}cdn(dM,0)=|\xi_{1}|\\
\xi_{1}\text{ contains }\xi_{2}  & \Leftrightarrow|\xi_{2}|=\tfrac{1}%
{2}cdn(dM,0)<|\xi_{1}|\text{ }\nonumber\\
\xi_{1}\text{ overlaps }\xi_{2}  & \Leftrightarrow1\leq\tfrac{1}%
{2}cdn(dM,0)<\min(|\xi_{1}|,|\xi_{2}|)\nonumber\\
\xi_{1}\text{ joint }\xi_{2}  & \Leftrightarrow cdn(dM,\leq2)=0\wedge
cdn(dM,4)>0\nonumber\\
\xi_{1}\text{ touch }\xi_{2}  & \Leftrightarrow cdn(dM,0)=0\wedge
cdn(dM,2)>0\nonumber\\
\xi_{1}\text{ disjoints }\xi_{2}  & \Leftrightarrow cdn(dM,\leq4)=0\nonumber
\end{align}

where $dM=dM(\xi_{1},\xi_{2})$, and $|\xi|$ is the cardinal number of $\xi$,
that is, if $\xi$ is a $m$-cell cluster, then $|\xi|=m$.

The third methods to determine Cs-Cs relations is directly using their
chromatic codes and distance $\delta(\xi_{1},\xi_{2})$, for example, using the
codes in Fig.7. Through tentative studies we found the general rule that `the
closer the chromatic distance, the closer the spatial topology' \cite{Zhu2010}%
. However, the full investigation and more rigorous mathematics remain for
future work.

\section{Discussions and summary}

As one type of spatial chromatic tessellations, full-OACDs provide a scheme to
partition and encode space. Technically a full-OACD is an irregular discrete
spatial data model based necessarily on the given object sets. If we have
enough perception, we should see that the similar schemes provide a new
approach to study discrete geometry. In SCT, as the generator number
increases, the cell number and neighborhood number of each cell will become
larger and larger, and the size of each cell will become smaller and smaller.
This property is different from other discrete tessellation models, such as
raster model and Voronoi diagrams. For example, the neighborhood number of a
pixel in a raster model is always 4 or 8, even though its spatial resolution
may be very high. When the generator number turns to be the infinite, the
space represented by SCT turns to be a continuous space, and its topology
turns to be the classic point-set topology.

Chromatic codes are the keys for characterization, computation and analysis of
spatial particles and their complexes. Spatial coding is a new topic in GIS.
As a scheme of spatial coding, full-OACD is still in its immature stage, where
many problems and directions remain unanswered and unexplored. Below we
discuss some issues that might be worth further investigation.

\textbf{Vertex types}. There are two types of vertexes in full-OACD: 2-I and
3-I vertexes, with quite different code bases. The 2-I's codes contain
half-integers but the 3-I's do not. 3-I vertexes actually are those
degenerated cells, also called singular cells in pervious studies
\cite{Zhu2010}. Instead of using perpendicular bisectors, if we use weighted
perpendicular bisectors, then 3-I vertexes will change their faces to real
cells, accompanying some new edges and 2-I vertexes, see an example in Fig.8a.
In such diagrams, particle bases are quite different from those in full-OACDs.

\begin{figure*}
\centerline{\includegraphics[width=5in]{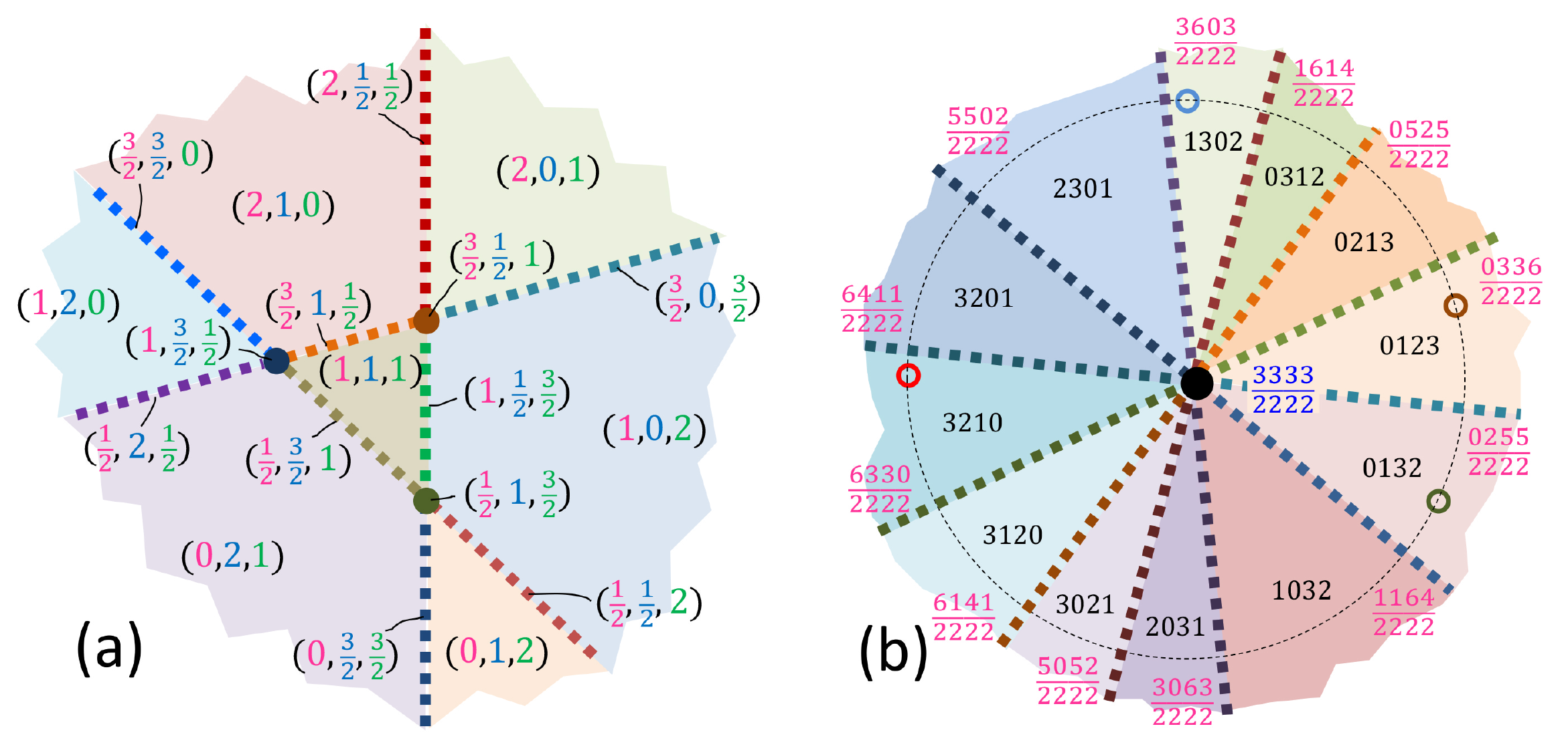}}
\caption{Other types of spatial chromatic tessellations. (a)
Diagrams generated by weighted bisectors; (b) Four generators are concyclic.}
\label{Fig8}
\end{figure*}

In addition, there are no more other types of vertexes in full-OACDs, such as
4-I or 5-I, since we have excluded them by assuming the generators being in
general cases. For non-general cases, for example, if 4 points are concyclic,
then their six perpendicular bisectors will intersect at the center of a
circle, and hence generate a 6-I vertex with code such as $(\frac{3}{2}%
,\frac{3}{2},\frac{3}{2},\frac{3}{2})$, see Fig.8b. Although we can exclude 4
concyclic points from the plane, in 3d space, 4 points are generally always on
a sphere, expect they are all in a plane, implying that many 6-I vertexes will
be found in $OACD(n,%
\mathbb{R}
^{3})$. This is similar to that 3 points are generally always concyclic in a
plane, except they are all in a line.

\textbf{Hidden spatial particles}. In terms of the Propositions
\ref{Proposition_EE_joint} and \ref{Proposition_CC_joint}, as well as the
results of functions $C2E$ and $E2V$, many spatial particles should exist in
full-OACDs but in one real space they are hidden. Checking the codes of the
hidden particles, we could not find any structural and permutation differences
from those emerged particles. However, the hidden particles will indeed emerge
in full-OACDs if we change point patterns of the generator set, but if we did
so, some previous emerged particles will be hidden again. Therefore a big
challenge of full-OACD is to determine what cells are hidden by given a kind
of generator pattern.

Hidden particles and complex codes can be also applied to analyze generator
patterns. For example, the complex of all 3-I vertexes in Fig.3b has a code
$(t_{1}^{21},t_{2}^{36},t_{3}^{43},t_{4}^{6},t_{5}^{19},t_{6}^{40})$,
indicating this complex is much closer to the generator $t_{3}=43$ than to
$t_{4}=6$. Note that in a SCM or full-OACD$(n,%
\mathbb{R}
^{n-1})$, all particles are emerged, but when it is mapped into the lower and
lower dimensional spaces, more and more particles will be hidden.

\textbf{From coded space to real space}. If we compare the belonging
relationships of different models proposed in SCM, the result should be that
OACD $\subseteq$ full-OACD $\subseteq$ SCT $\subseteq$ SCM. Full-OACD is
generated from half-space partitions, so it is still a type of spatial
chromatic tessellations (SCT), namely, a mapping from SCM spaces, which are
usually in higher dimensions, to real spaces, which are usually in lower
dimensions. From the perspective of SCM, there is a question that how to
understand the half-integers in edge codes. If we are allowed to use
half-integers to make space, then how about if we use other numbers such as
$\frac{1}{4}$-integers. It seems that edges and vertexes are not real spaces,
just as we often say that lines are only with lengths but no areas. When we
use a pen to draw a line to partition a piece of paper, it appears that we get
three parts: half at right, half at left, and the middle line that cannot be
assigned to any half. When we use a scissors to cut a piece of paper, however,
we can only get two pieces of paper but never three parts. Therefore, we would
like to emphasize that in full-OACD and SCM, the cell is the elementary
subspace, whereas edges, vertexes and other lower dimensional subspaces are
only boundaries. Therefore, in order to use coded spaces to represent the
real-world spaces, we suggest only use cells. It is like in raster model, any
spatial entities and objects are always represented by pixels, no matter it is
a point, line, or area.

\end{document}